\documentclass[11pt,letterpaper]{article}
\usepackage[utf8]{inputenc}

\usepackage{amsmath,amssymb,amsfonts} 
\usepackage{epsfig} \usepackage{latexsym,nicefrac,bbm}
\usepackage{xspace}
\usepackage{color,fancybox,graphicx,subfigure,fullpage}
\usepackage[top=1in, bottom=1in, left=1in, right=1in]{geometry}
\usepackage{tabularx} \usepackage{hyperref} 
\usepackage{pdfsync}
\usepackage[boxruled, linesnumbered]{algorithm2e}
\usepackage{multicol}
\usepackage{enumitem}
\usepackage{multirow}
 \usepackage{url}

\renewcommand{\epsilon}{\varepsilon}

\newcommand{\eps}{\varepsilon}

\newtheorem{theorem}{Theorem}[section]

\newtheorem{definition}{Definition}[section]
\newtheorem{property}{Property}[section]
\newtheorem{lemma}[theorem]{Lemma}

\newtheorem{proposition}[theorem]{Proposition}
\newtheorem{corollary}[theorem]{Corollary}

\newenvironment{proof}{\begin{trivlist} \item {\bf Proof:~~}}
   {\qed\end{trivlist}}

\def\FullBox{\hbox{\vrule width 6pt height 6pt depth 0pt}}

\def\qed{\ifmmode\qquad\FullBox\else{\unskip\nobreak\hfil
\penalty50\hskip1em\null\nobreak\hfil\FullBox
\parfillskip=0pt\finalhyphendemerits=0\endgraf}\fi}

\usepackage[T1]{fontenc}
\usepackage{lmodern}

\SetKwInOut{KwHyperparameters}{Hyperparameters}

\makeatletter
\def\blfootnote{\xdef\@thefnmark{}\@footnotetext}
\makeatother

\title{\bf  Sampling from Log-Concave Distributions over Polytopes
via a Soft-Threshold Dikin Walk}

 \author{Oren Mangoubi\\ Worcester Polytechnic Institute \and Nisheeth K. Vishnoi \\ Yale University}

\begin{document}
 \date{}
\maketitle

\begin{abstract}
Given a Lipschitz or smooth convex function $f:K \to \mathbb{R}^d$ for a bounded polytope $K$ defined by $m$ inequalities, we consider the problem of sampling from the log-concave distribution $\pi(\theta) \propto e^{-f(\theta)}$ constrained to $K$.
Interest in this problem derives from its applications to  Bayesian inference and differentially private learning.
Our main result is  a generalization of the Dikin walk Markov chain to this setting that requires at most
$O((md + d L^2 R^2) \times md^{\omega-1}) \log(\frac{w}{\delta}))$ arithmetic operations to sample from $\pi$ within error $\delta>0$ in the total variation distance from a $w$-warm start. 
Here $L$ is the Lipschitz-constant of $f$, $K$ is contained in a ball of radius $R$ and contains a ball of smaller radius $r$, and $\omega$ is the matrix-multiplication constant.
Our algorithm improves on the running time of prior works for a range of parameter settings important for the aforementioned learning applications. 
Technically, we depart from previous Dikin walks by adding a ``soft-threshold'' regularizer derived from the Lipschitz or smoothness properties of $f$  to the log-barrier function for $K$ that allows our version of the Dikin walk to propose updates  that have a high Metropolis acceptance ratio for $f$, while at the same time remaining inside the polytope $K$.\blfootnote{We originally introduced the soft-threshold Dikin walk in the first arXiv version of  \cite{OM_Sampling22}.
Subsequently, we realized that this sampling algorithm is of independent interest with further applications. 
Thus, we have since removed the results for our soft-threshold Dikin walk Markov chain in the latest version of \cite{OM_Sampling22}, and instead, present them separately here in an expanded manner.}

\end{abstract}

\thispagestyle{empty}

\newpage

\tableofcontents

\thispagestyle{empty}

 \newpage

\setcounter{page}{1}

\section{Introduction}
We consider the  problem of sampling from a log-concave distribution supported on a polytope:
Given a polytope 
$K:=\{\theta \in \mathbb{R}^d: A\theta \leq b\}$,
where $A\in   \mathbb{R}^{m\times d}$ and $b \in \mathbb{R}^m$,
and a convex function $f:K \rightarrow \mathbb{R}$, output a sample $\theta \in K$ from the distribution 
$\pi(\theta) \propto e^{-f(\theta)}$.

Two special cases of this problem, sampling from a uniform distribution over a polytope ($f \equiv 0$), and sampling from an unconstrained log-concave distribution (when $K=\mathbb{R}^n$), have been classically studied  in various disciplines. 
Sampling exactly in both these  cases is hard in general and a large body of work has developed ``approximate'' sampling algorithms under a variety of conditions on $f$ and $K$.
The problem of sampling from the uniform distribution  when $K$ is a polytope given by a set of inequalities, and where the distance to $\pi$ is measured in the total-variation (TV) distance has been widely studied; see, e.g., \cite{kannan2012random, narayanan2016randomized, sachdeva2016mixing, lee2018convergence, chen2017vaidya, laddha2020strong}.
For the {unconstrained} log-concave sampling problem, one line of work requires the log-density to be $L$-Lipschitz or $\beta$-smooth over {\em all} of $\mathbb{R}^d$ for some $L, \beta>0$, and which give bounds on the distance to the target density $\pi$ in terms of the TV distance \cite{dwivedi2018log}, Wasserstein distance \cite{durmus2017nonasymptotic, dalalyan2020sampling}, and Kullback-Leibler (KL) divergence \cite{wibisono2018sampling, durmus2019analysis}.

Several works  provide algorithms for sampling in the more general setting  when $K$ is an arbitrary convex body given by a membership oracle
\cite{frieze1994sampling,applegate1991sampling,  frieze1999log,lovasz2006fast,lovasz2007geometry}. 
In particular,  \cite{lovasz2006fast} gives an  algorithm that can sample from a log-concave distribution $\pi \propto e^{-f}$ on a convex body $K$ which contains a ball of radius $r$ with TV error $\delta>0$ in $O\left(d^{4.5} \log^5\frac{d}{\delta r}\right)$ calls to a membership oracle for $K$. 
 Further, \cite{bubeck2015finite,brosse2017sampling} provide versions of the Langevin dynamics  for sampling from a log-concave distribution $\pi \propto e^{-f}$ on $K$ where $f$ is $L$-Lipschitz and $\beta$-smooth and $K$ is given by a projection oracle-- when, e.g., $K$ is contained in a ball of radius $O(1)$, it contains a ball of radius $\Omega(1)$, and   $L=\beta =O(1)$,    \cite{bubeck2015finite} give a bound of roughly $O(\frac{d^9}{\delta^{22}})$ gradient and projection oracle calls to sample from $\pi$ with TV error $\delta>0$, while \cite{brosse2017sampling} give a bound of roughly  $O(\frac{d^5}{\delta^6})$ gradient and projection oracle calls.
Moreover, \cite{gopi2022private} provide an algorithm for sampling from distributions $\propto e^{-f}$ where $f$ is both  $L$-Lipschitz and $\mu$-strongly convex (that is, $f(\theta)$ is the sum of a convex function $\hat{f}(\theta)$ and the quadratic function $\frac{\mu}{2}\|\theta\|^2$ for some $\mu>0$) on a convex body $K$ in roughly $\tilde{O}(\frac{L^2}{\mu}\log^2(\frac{d}{\delta}))$ membership and function oracle calls.
While their runtime is logarithmic in dimension and polynomial in the strong convexity and Lipschitz parameters of $f$, their algorithm does not apply to the more general  setting (of interest here) where $f$ is convex and $L$-Lipschitz (or $\beta$-smooth), but not necessarily $\mu$-strongly convex for any $\mu>0$, which includes, e.g.,  sampling problems arising when training Lasso logistic regression models under differential privacy. 

Our interest in  designing  algorithms for the polytope-constrained log-concave sampling problem derives from its applications to  areas  such as Bayesian inference and differentially private optimization. 
In Bayesian inference, the ability to sample from  $\pi(\theta) \propto e^{-f(\theta)}$  allows one to compute Bayesian confidence intervals and other statistics for the Bayesian posterior distribution of many machine learning models (see e.g. \cite{gamerman2006markov, gelman1996markov, carlin1995bayesian}).
In differentially private optimization, sampling from the ``exponential mechanism'' \cite{mcsherry2007mechanism} allows one to achieve optimal utility bounds for the problem of minimizing $f$ under $\varepsilon$-differential privacy \cite{bassily2014private}.

The instances of the sampling problem that arise in these applications, while more general than the two special cases ($f \equiv 0$ and $K=\mathbb{R}^d$), have more structure than the case of a general log-concave function supported on an arbitrary convex body. 
For instance, in the Bayesian Lasso logistic regression problem, the  function $f$ is $f(\theta) = \sum_{i=1}^n \ell(\theta; x_i)$, where $\ell$ is the logistic loss 
and $x_i$ are the datapoints with $\|x_i\|_2 \leq 1$.
The constraint polytope is $K = \{\theta \in \mathbb{R}^d:  \|\theta\|_1 \leq O(1) \}$; see \cite{tian2008efficient, silvapulle2011constrained, kim2018logistic,BayesianLasso}. 
Note that, since the logistic function is both $O(1)$-smooth and $O(1)$-Lipschitz, $f$ is both $O(n)$-Lipschitz and $O(n)$-smooth, and $K$ is defined by $2d$ inequalities and contained in a ball of radius $O(1)$.
%
To obtain an $\eps$-differentially private mechanism for the Lasso logistic regression problem, using the exponential mechanism of \cite{mcsherry2007mechanism}, the goal is to sample from $e^{-\frac{\epsilon}{R} \sum_{i=1}^n \ell(\theta; x_i)}$, where $\ell$ is the logistic loss and $K$ is contained in a ball of radius $R$.
Thus, the log-density is both $\beta$-smooth and  $L$-Lipschitz  for $\beta = L =   \frac{n\epsilon}{R} = O(d)$ if $n=d$ and $\epsilon <1$ since $R=O(1)$. 

Another example,  is a result of \cite{leake2020polynomial} (see also \cite{Ge}) that reduces the problem of $\epsilon$-differentially private low-rank approximation of a symmetric $p \times p$ matrix to a constrained sampling problem with $f$ being a linear function (which is trivially $0$-smooth) of dimension $d=p^2$
and $K$ being the Gelfand-Tsetlin polytope (a generalization of the probability simplex). 
$K$
has $d$ inequalities and diameter $ \sqrt{d}$; see \cite{leake2020polynomial}.

When sampling from the exponential mechanism in privacy applications, sampling with total variation bounds is insufficient to guarantee $\epsilon$-differential privacy, the strongest notion of differential privacy; see \cite{dwork2014algorithmic}.
Instead, one requires bounds in the stronger infinity-distance metric $\mathrm{d}_\infty(\nu, \pi):= \sup_{\theta \in K} \left|\log  \frac{\nu(\theta)}{\pi(\theta)}\right|$; see e.g.  \cite{dwork2014algorithmic}. 
A recent work \cite{OM_Sampling22} showed how to convert samples within $O(\delta)$-TV distance from continuous-space log-Lipschitz densities $\pi$, into samples with $O(\epsilon)$-infinity-distance bounds, but it requires the TV distance $\delta$ to be very small--roughly $\delta = O(\epsilon e^{-d -LR})$.
Thus, continuous-space Markov chains whose  runtime bounds have a high-order dependence on $\log\frac{1}{\delta}$ incur an additional factor in their runtime which has a high-order {\em polynomial} dependence on the dimension $d$, when applied to privacy problems where infinity-distance bounds are required; for instance, the runtime of the hit-and-run Markov chain \cite{lovasz2006fast} of $O\left(d^{4.5} \log^5\frac{d}{\delta r}\right)$ membership oracle calls for $O(\delta)$-TV sampling leads to a runtime of roughly $\tilde{O}\left(d^{9.5} + d^{4.5}(LR)^5\right)$ membership oracle calls when $\epsilon$-infinity distance bounds are required.
Thus, for applications in differential privacy, it is desirable to design sampling algorithms that have a low-order polynomial dependence not only on the parameters $d, L,R, \beta$ but also  on  $\log\frac{1}{\delta}$.

The most relevant work towards the problem of designing fast algorithms for sampling from Lipschitz or smooth log-densities supported on explicitly given polytopes  is by \cite{narayanan2017efficient} who extends the (Gaussian) Dikin walk to allow it to sample from any log-concave distribution $\pi \propto e^{-f}$ on $K$ where $f$ is $L$- Lipschitz or $\beta$-smooth.
 The Dikin walk Markov chain was  introduced in  \cite{kannan2012random} in the special case where $f \equiv 0$ (see also  \cite{narayanan2016randomized, sachdeva2016mixing}). 
From any $\theta$ in the interior of $K$, the version of the Dikin walk in \cite{kannan2012random}  proposes a random step $z$ uniformly distributed on (a scalar multiple of) the Dikin ellipsoid  $E(\theta) := \{w: w^\top H^{-1}(\theta) w \leq 1\}$, where 
$H(\theta) :=  \sum_{j=1}^{m} \frac{a_j a_j^\top}{(b_j - a_j^\top \theta)^2}$
is the Hessian of the log-barrier function $\varphi(\theta) := -  \sum_{j=1}^{m} \log(b_j-a_j^\top \theta )$ at $\theta$.
Here, $a_j$ is the $j$th row of $A$.
As the Dikin ellipsoid  $E(\theta)$ is a subset of $K$ for any $\theta$ in the interior of $K$, the Dikin walk remains inside $K$.
To ensure that the Dikin walk samples from the uniform distribution on $K$, the proposed step is accepted with probability determined by a Metropolis acceptance rule.
\cite{narayanan2016randomized, sachdeva2016mixing} analyze a version of the Dikin Walk Markov chain where the proposed step is instead sampled from a Gaussian distribution with covariance matrix determined by the Dikin ellipsoid.

\cite{narayanan2017efficient} introduces an additional coefficient, $\exp\left(\frac{f(\theta)}{f(z)}\right)$, to the Metropolis acceptance probability which ensures that the stationary distribution of the Dikin walk Markov chain is the target distribution $\pi$. %
 To ensure that the proposed step $z$ is accepted with high probability, $z$ is sampled from a Gaussian distribution with covariance matrix $\gamma^2 H^{-1}(\theta)$, where $\gamma := \min(\frac{1}{d}, \frac{1}{L})$ (if $f$ is $L$-Lipschitz) or $\gamma := \min(\frac{1}{d}, \frac{1}{\beta})$  if $f$ is $\beta$-smooth) is a scalar coefficient.
The runtime bounds in \cite{narayanan2017efficient} imply that their Dikin walk takes $O((md^{4+\omega} + md^{2+\omega}L^2R^2)\log(\frac{w}{\delta}))$ arithmetic operations to sample within TV error $O(\delta)$ from an $L$-log Lipschitz density on a polytope from a $w$-warm start.
A distribution $\nu$ is $w$-warm for $w \geq 1$ with respect to the stationary distribution $\pi$ if $\sup_{z\in K} \frac{\nu(z)}{\pi(z)} \leq w$, and $\omega$ denotes the matrix-multiplication constant.

\section{Our results}
Our main result is a sampling algorithm inspired by the Dikin walk Markov chain whose steps are determined by a barrier function that generalizes the log-barrier function by adding a ``soft-threshold'' regularizer  (Algorithm \ref{alg_Soft_Dikin_Walk} and Theorem \ref{thm_soft_threshold_Dikin}). 
The algorithm  generates samples from an  $L$-Lipschitz or $\beta$-smooth log-density on an $R$-bounded polytope with an error bounded in the TV distance.
Theorem \ref{thm_soft_threshold_Dikin} often results in the fastest known) algorithm for some of the applications to Bayesian inference and differentially private optimization mentioned in the introduction. 

\medskip
In the following, for any two probability distributions $\mu, \nu$ on $\mathbb{R}^d$, we denote the total variation distance between $\mu$ and $\nu$ by $\|\mu - \nu \|_{\mathrm{TV}}:= \sup_{S \subseteq \mathbb{R}^d} |\mu(S) - \nu(S)|$.
For any $\theta \in \mathbb{R}^d$ and $t>0$ we denote the ball of radius $t$ at $\theta$ by $B(\theta,t):= \{z \in \mathbb{R}^d: \|z-\theta\|_2\leq t\}$ where $\|\cdot\|_2$ denotes the Euclidean norm.
For any subset $S\subset \mathbb{R}^d$ we denote the interior of $S$ by $\mathrm{Int}(S):= \{\theta \in S :  B(\theta, t) \subseteq S \textrm{ for some } t>0\}$.

\begin{theorem}[Sampling with TV bounds via a soft-threshold Dikin Walk]\label{thm_soft_threshold_Dikin}
There exists an algorithm (Algorithm \ref{alg_Soft_Dikin_Walk}) which, given $\delta, R>0$ and either $L>0$ or $\beta>0$,  
$A \in \mathbb{R}^{m \times d}$, $b\in \mathbb{R}^m$ that define a polytope $K := \{\theta \in \mathbb{R}^d : A \theta \leq b\}$  such that 
        $K$ is contained in a ball of radius $R$ and has nonempty interior,
         an oracle for the value of a convex function $f: K \rightarrow \mathbb{R}^d$, where $f$ is either $L$-Lipschitz or $\beta$-smooth, 
        and an initial point sampled from a distribution supported on $K$ which is $w$-warm with respect to $\pi \propto e^{-f}$ for some $w>0$,
        outputs a point from a distribution $\mu$ where $\|\mu- \pi\|_{\mathrm{TV}} \leq \delta$.
        Moreover, this algorithm takes at most $O((md + d L^2 R^2) \times \log(\frac{w}{\delta}))$ Markov chain steps in the setting where $f$ is $L$-Lipschitz, or $O((md + d \beta R^2) \times \log(\frac{w}{\delta}))$ Markov chain steps in the setting where $f$ is $\beta$-smooth, where each step makes one function evaluation and $O(md^{\omega-1})$ arithmetic operations.
\end{theorem}

\noindent
In particular, Theorem \ref{thm_soft_threshold_Dikin} improves on the previous bounds of $O((md^{4+\omega} + md^{2+\omega}L^2R^2)\log(\frac{w}{\delta}))$ arithmetic operations to sample from an $L$-log-Lipschitz density on a polytope with TV error $O(\delta)$ implied by the work of \cite{narayanan2017efficient} for a different version of the Dikin walk algorithm, by a factor of (at least) $\min(\frac{1}{m}d^4,\, d^2)$. \footnote{When setting their scalar step size hyperparameter $\alpha^{\frac{1}{2}}$ to  $O(\min(\frac{1}{d}, \frac{1}{LR}))$,  \cite{narayanan2017efficient} obtain a bound of $\phi \geq \frac{\alpha^{\frac{1}{2}}}{\kappa \sqrt{d}}$ on the conductance $\phi$ of their  Dikin walk (Lemma 4 of their paper).
Here $\kappa$ is the self-concordance parameter of the barrier function; for the log-barrier  $\kappa = m$, although there are other barrier functions for which $\kappa=O(d)$.
Plugging their conductance bound into Corollary 1.5 of \cite{lovasz1993random} implies a bound of roughly $\phi^{-2} \log(\frac{w}{\delta}) = O((d^5 + d^3L^2R^2)\log(\frac{w}{\delta}))$ steps from a $w$-warm start for their Dikin walk to sample with $O(\delta)$ TV error from $\pi$, implying a bound of $O((md^{4+\omega} + md^{2+\omega}L^2R^2)\log(\frac{w}{\delta}))$ arithmetic operations.}
If $m=O(d)$, the improvement is $d^2$, and if we also have $LR = O(\sqrt{d})$, the improvement is $d^3$.
Thus, e.g., in the example of Bayesian Lasso logistic regression, the runtime is $O((md + d \beta R^2) \log(\frac{w}{\delta})) = O(d^2 \log(\frac{w}{\delta}))$ Dikin walk steps from a $w$-warm start since $m=d$ and $n=d$, since $\beta = n=d$ and $R=O(1)$.
We note that while many works, e.g. \cite{kannan2012random, narayanan2016randomized, sachdeva2016mixing, laddha2020strong, chen2017vaidya}, provide faster bounds for the Dikin walk and its variants than the ones given in \cite{narayanan2017efficient}, these faster bounds only apply in the special case when $\pi$ is the uniform distribution on $K$.

  Moreover, when a warm start is not provided, Algorithm \ref{alg_Soft_Dikin_Walk} takes at most $O((md + d L^2 R^2) \times (d\log(\frac{R}{r})+M +\log(\frac{1}{\delta})))$ Markov chain steps when $f$ is $L$-Lipschitz (or $O((md + d \beta R^2) \times (d\log(\frac{R}{r})+\log(M) +\log(\frac{1}{\delta})))$ steps when $f$ is $\beta$-smooth), since an $e^{d\log(\frac{R}{r})+M}$-warm start can be obtained by sampling uniformly from the ball  $B(a,r)$ contained in $K$, where $M = \log\left(\frac{\max_{\theta \in K}e^{-f(\theta)}}{e^{-f(a)}}\right) \leq e^{LR}$. 
     In comparison, the work of \cite{lovasz2006fast} implies a bound of $O(m d^{5.5} \log^5(\frac{d}{\delta r}))$ arithmetic operations to sample with TV error $O(\delta)$ in the setting where $f$ is $L$-Lipschitz and constrained to a polytope, regardless of whether a warm start is provided.\footnote{Corollary 1.2 in \cite{lovasz2006fast}, together with the rounding procedure in their optimization algorithm in Section 5 of their paper, imply a bound of roughly $O(d^{4.5} (\log^5(\frac{d}{\delta r}))$ hit-and-run Markov chain steps, where $M := \log(\max_{\theta \in K}e^{f(\theta_0) -f(\theta)})$, to sample within TV distance $O(\delta)$ from a log-concave distribution on a convex body contained in a ball of radius $R$ and containing a ball of smaller radius $r$.
If $f$ is constrained to a polytope given by $m$ inequalities, each step requires $md$ operations to compute a membership oracle. 
Thus,  if $f$ is also $L$-Lipschitz   and, e.g.,  $LR=O(\sqrt{d})$ and $m=O(d)$, the bound implied by  \cite{lovasz2006fast}  is $O(d^{6.5}\times \log^5(\frac{d R L}{\delta r}))$ arithmetic operations.} 
  Thus,  in the setting where $\pi$ is constrained to a polytope $K$, Theorem \ref{thm_soft_threshold_Dikin}  improves on the bounds of \cite{lovasz2006fast} for the hit-and-run algorithm by a factor of roughly $\frac{1}{m}d^{4.5- \omega}\log^4(\frac{1}{\delta})$ when $LR=O(\sqrt{d})$.
  If $m=O(d)$, the improvement is $d^{3.5- \omega}\log^4(\frac{1}{\delta})$.
On the other hand, we note that \cite{lovasz2006fast} applies more generally when $\pi$ is a log-concave distribution on a convex body and their dependence on $R$ is logarithmic, while our bounds for the soft-threshold Dikin walk (Theorem \ref{thm_soft_threshold_Dikin})  apply to the setting where $\pi$ is a log-Lipschitz log-concave distribution on a polytope.
In the example of Bayesian Lasso logistic regression, our algorithm takes $O(d^{3+\omega} \log(\frac{R}{r \delta})$ arithmetic operations since  $f$ is both $\beta$-smooth and $L$-Lipschitz, with $\beta = L = n =m =d$,  and $R=O(1)$; this improves by a factor of $d^{3.5 - \omega}  \log^4(\frac{dR}{\delta r}) > d \log^4(\frac{dR}{\delta r}) $ on the bound of $d^{6.5} \log^5(\frac{dR}{\delta r})$ arithmetic operations  for the hit-and-run algorithm of \cite{lovasz2006fast}.

The proof of Theorem \ref{thm_soft_threshold_Dikin} appears in Section \ref{sec_mainproof}.
We present an overview of the main ideas in the proof of Theorem \ref{thm_soft_threshold_Dikin} in Section \ref{sec_technical_overview}.
In Section \ref{sec:barrier}, we give an axiomatic approach to arrive at our barrier function and discuss possible extensions.

\paragraph{Detailed comparison to \cite{narayanan2017efficient}.} Compared to \cite{narayanan2017efficient}, in our version of the Dikin walk, $z$ is sampled from a Gaussian distribution with covariance matrix $$(\alpha^{-1} H(z) +  \eta^{-1}I_d)^{-1},$$ where $\alpha^{-1}$ and  $\eta^{-1}$  are hyper-parameters chosen to be $\alpha^{-1} \approx d^2$ and $\eta^{-1} \approx d L^2$ if $f$ is $L$-Lipschitz (or $\eta^{-1} \approx d \beta$  if $f$ is $L$-smooth).
 The ``soft-threshold'' regularization term $\eta^{-1}I_d$ prevents the Markov chain from taking steps where the value of $f$ decreases by more than $O(1)$ w.h.p., ensuring that the term $e^{f(\theta)-f(z)}$ in the acceptance probability is $\Omega(1)$.
The soft-threshold regularizer is chosen to be a multiple of the identity matrix $I_d$ since the Lipschitz condition on $f$ is rotationally invariant-- it bounds the derivative of $f$ by the same amount $L$ in each direction (the same is true for the second derivative of $f$ if $f$ is $\beta$-smooth).
This in turn allows our choice of scaling $\alpha^{-1}$ for the ``Dikin ellipsoid'' term $H(z)$-- which is not in general rotationally invariant, and determined only by the geometry of the polytope rather than the geometry of the function $f$-- to be independent of $L$ (or $\beta$).
This is in contrast to the Dikin walk in \cite{narayanan2017efficient} where the scaling parameter for $H(z)$ must depend on $L$ (or $\beta)$ to ensure an $\Omega(1)$ acceptance probability, which allows our Markov chain to propose steps with a larger variance than the Dikin walk in \cite{narayanan2017efficient} in directions which are not the largest eigenvector of $H^{-1}(z)$.

The (inverse) covariance matrix $\alpha^{-1} H(z) +  \eta^{-1}I_d$ of our soft-threshold Dikin walk updates is the Hessian of the function $\psi(\theta) = \alpha^{-1} \varphi(\theta) + \eta^{-1}\|\theta\|_2^2$.
This barrier function can be seen to be a Hessian of a self-concordant barrier function.  
On the other hand, it is not the Hessian of a logarithmic-barrier function for any polytope defined by any set of inequalities.
This prevents us from directly applying the analysis of the Dikin walk in the special case where $f \equiv 0$ \cite{narayanan2016randomized, sachdeva2016mixing}--which relies on properties of log-barrier functions-- to our soft-threshold Dikin walk on Lipschitz or smooth $f$.
To get around this problem, we show that, while $\psi(\theta)$ is not a log-barrier function of any polytope $K$, it is the limit of a sequence of log-barrier functions $\hat{\psi}_1, \hat{\psi}_2, \ldots$ where $\hat{\psi}_i(\theta) \rightarrow \psi(\theta)$ uniformly in $x$ as $i \rightarrow \infty$.
See  Section \ref{sec_technical_overview} for a detailed overview of the proof.

 An open problem is  to obtain runtime bounds for the Dikin walk which  do not require $f$ to be Lipschitz or smooth and/or that depend polynomially on $\log R$.
  This leads to the related  question of whether one can design other tractable  self-concordant barrier functions   to obtain further improvements in the runtime for sampling from log-concave distributions $\propto e^{-f}$ on a polytope $K$.
We discuss possible extensions in Section \ref{sec:barrier}.

\vspace{-2mm}

\paragraph{Infinity-distance sampling.} 
In applications of sampling to differentially private optimization \cite{mcsherry2007mechanism, hardt2010geometry, bassily2014private,ganesh2020faster, leake2020polynomial},
bounds in the total variation (TV) distance are insufficient to guarantee ``pure''  differential privacy, and 
one instead requires bounds in the infinity-distance $\mathrm{d}_\infty(\nu, \pi):= \sup_{\theta \in K} |\log  \frac{\nu(\theta)}{\pi(\theta)}|$; see e.g.  \cite{dwork2014algorithmic}.
\cite{OM_Sampling22} gives an algorithm that converts samples from TV bounds to those bounded 
in the infinity-distance. 
Namely, given any $\varepsilon>0$ and a sample from a distribution $\mu$ within TV distance 
 $\delta \leq O\left(\epsilon\times \left(\frac{R(d \log(\nicefrac{R}{r})+LR)^2}{\eps r}\right)^{-d} e^{-LR}\right)$ of $\pi$, this post-processing algorithm 
outputs a sample from a distribution $\nu$ with infinity  $\mathrm{d}_\infty(\nu, \pi) \leq \varepsilon$ from $\pi$.
Plugging  the TV bounds from our Theorem \ref{thm_soft_threshold_Dikin} into Theorem 2.2 of \cite{OM_Sampling22} gives a faster algorithm to sample from a log-concave and log-Lipschitz distribution constrained to a polytope $K$, with $O(\eps)$ error in $\mathrm{d}_\infty$.
In particular, Corollary \ref{thm_infinity_divergence_sampler} improves the bound of $O\left((m^2d^3 + m^2 d L^2 R^2) \times \left[LR + d\log\left(\frac{Rd +LRd}{r \eps}\right)\right]\right)$
in  Theorem 2.1 of \cite{OM_Sampling22}.
The proof is identical to that of how Theorem 2.2  implies Theorem 2.1 in \cite{OM_Sampling22}, and we refer the reader to \cite{OM_Sampling22}.

\begin{corollary}[Log-concave sampling on a polytope with infinity-distance guarantees]\label{thm_infinity_divergence_sampler}
There exists an algorithm which, given $\epsilon, L, r, R>0$,  $A \in \mathbb{R}^{m \times d}$, $b\in \mathbb{R}^m$ (and possibly $\beta>0$), that define a polytope $K := \{\theta \in \mathbb{R}^d : A \theta \leq b\}$  contained in a ball of radius $R$, a point $a \in \mathbb{R}^d$ such that $K$ contains a ball $B(a,r)$ of smaller radius $r$, and an oracle for the value of a convex function $f: K \rightarrow \mathbb{R}^d$, where $f$ is  $L$-Lipschitz (or is both $L$-Lipschitz  and $\beta$-smooth), and defining $\pi$ to be the distribution $\pi \propto e^{-f}$, outputs a point from a distribution $\nu$ such that $\mathrm{d}_\infty(\nu, \pi)< \eps$.
        Moreover, with very high probability\footnote{The number of steps is $O(\tau \times T)$,  where $\mathbb{E}[\tau] \leq 3$, $\mathbb{P}(\tau \geq t) \leq \left(\frac{2}{3}\right)^t$ for $t \geq 0$, and $\tau \leq O(d\log(\frac{R}{r}) + LR)$ w.p. 1.}, this algorithm  takes $O(T)$ function evaluations and $O(T\times md^{\omega-1})$ arithmetic operations, where $T=O\left((md + d L^2 R^2) \times \left[LR + d\log\left(\frac{Rd +LRd}{r \eps}\right)\right]\right)$ if $f$ is $L$-Lipschitz, or $T=O\left((md + d \beta R^2) \times \left[LR + d\log\left(\frac{Rd +LRd}{r \eps}\right)\right]\right)$ if $f$ is also $\beta$-Lipschitz. 

\end{corollary}

\noindent
Corollary \ref{thm_infinity_divergence_sampler} further improves the dependence on the dimension $d$ over \cite{bassily2014private, OM_Sampling22}. %
Specifically, when each function evaluation takes $O(d^2)$ arithmetic operations (this is the case, e.g., in the setting where evaluating $f$ requires computing at most $d$ $d$-dimensional inner products), 
\cite{bassily2014private} implies a bound of   $O\left(\frac{1}{\eps^2}(md^{11} + md^7L^4R^4)\times \mathrm{polylog}\left(\frac{1}{\eps},\frac{1}{r}, R,L,d\right)\right)$ arithmetic operations. 
 Corollary \ref{thm_infinity_divergence_sampler} improves on this bound by a factor of roughly $\frac{1}{\epsilon^2 m}d^{10-\omega}$.  For example, when $m=O(d)$, as may be the case in  privacy applications, the improvement is   $\frac{1}{\epsilon^2}d^{9-\omega}$.
 Moreover, it improves by a factor of $d^3$ on the bound of $O\left((m^2d^3 + m^2 d L^2 R^2) \times \left[LR + d\log\left(\frac{Rd +LRd}{r \eps}\right)\right] \times md^{\omega-1}\right)$ (or  $O\left((m^2d^3 + m^2 d \beta R^2) \times \left[LR + d\log\left(\frac{Rd +LRd}{r \eps}\right)\right] \times md^{\omega-1}\right)$ if $f$ is also $\beta$-smooth) arithmetic operations obtained in \cite{OM_Sampling22} by plugging in the Dikin walk of \cite{narayanan2017efficient} into   \cite{OM_Sampling22}.

\paragraph{Differentially private optimization.}
A randomized mechanism $h: \mathcal{D}^n \rightarrow \mathcal{R}$ is  said to be  $\epsilon$-differentially private if for any datasets $x, x'  \in \mathcal{D}$ which differ by a single datapoint, and any $S \subseteq \mathcal{R}$, we have that $\mathbb{P}(h(x) \in S) \leq e^{\epsilon} \mathbb{P}(h(x') \in S);$ see \cite{dwork2014algorithmic}.
In the application of the exponential mechanism to $\epsilon$-differentially private low-rank approximation of a $p\times p$ symmetric matrix $M$ \cite{leake2020polynomial} (see also \cite{Ge}), one wishes to sample within infinity distance $O(\epsilon)$ from a log-linear distribution $\propto e^{-f}$ on the Gelfand-Tsetlin polytope  $K \subseteq \mathbb{R}^d$  (which generalizes the probability simplex), where $d= p^2$, and where $K$ has
 $m=d$ inequalities with diameter $R= O(\sqrt{d})$.
In this application, the log-linear density $f$ is  (trivially) $0$-smooth and $d^2 \sigma_1$-Lipschitz, where $\sigma_1 := \|M\|_2$ is the spectral norm of $M$.
Thus, when applied to the mechanism of \cite{leake2020polynomial}, our algorithm takes $d^{4.5+\omega}\sigma_1 \log(\frac{1}{\epsilon})$ arithmetic operations.
This improves by a factor of $d^3$ on the runtime bound of $O(d^{7.5 + \omega} \sigma_1)$ arithmetic operations implied by \cite{OM_Sampling22, narayanan2017efficient}, and improves by a factor of $\frac{1}{\epsilon^2}d^{11.5} \sigma_1^3$ on the bound of $O\left(\frac{1}{\epsilon^2}d^{16}\sigma_1^4\right)$ arithmetic operations for the bound in \cite{bassily2014private}.

Consider the problem of finding an (approximate) minimum $\hat{\theta}$ of an empirical risk function $f:  K  \times \mathcal{D}^n \rightarrow \mathbb{R}$ under the constraint that the output $\hat{\theta}$ is $\epsilon$-differentially private,   where  $f(\theta, x) := \sum_{i=1}^n \ell_i(\theta,x_i)$.
We assume that the $\ell_i(\cdot, x)$ are $\hat{L}$-Lipschitz for all $x\in \mathcal{D}^n$, $i \in \mathbb{N}$, for some given $\hat{L}>0$.
In this setting, \cite{bassily2014private} show that the minimum ERM utility bound under the constraint that $\hat{\theta}$ is pure $\epsilon$-differentially private,   $\mathbb{E}_{\hat{\theta}}[f(\hat{\theta},x)] - \min_{\theta \in K} f(\theta,x) = \Theta(\frac{d \hat{L} R}{\epsilon})$,
 is achieved if one samples $\hat{\theta}$ from the exponential mechanism $\pi \propto e^{- \frac{\eps}{2\hat{L}R}f}$ with infinity-distance error at most $O(\eps)$.
Plugging Corollary \ref{thm_infinity_divergence_sampler} into the framework of the exponential mechanism, we obtain a faster algorithm for pure $\eps$-differentially private mechanism which achieves the minimum expected risk (Corollary \ref{corr_DP}).

\begin{corollary}[Differentially private empirical risk minimization]\label{corr_DP}
There exists an\, algorithm which, given $\epsilon, \hat{L}, r, R>0$,  $A \in \mathbb{R}^{m \times d}$, $b\in \mathbb{R}^m$ (and possibly $\hat{\beta}>0$) that define a polytope $K := \{\theta \in \mathbb{R}^d : A \theta \leq b\}$  contained in a ball of radius $R$, a point $a \in \mathbb{R}^d$ such that $K$ contains a ball $B(a,r)$ of smaller radius $r$,
and an empirical risk function $f(\theta, x) := \sum_{i=1}^n \ell_i(\theta,x_i)$, where each $\ell_i: K \rightarrow \mathbb{R}$ is $\hat{L}$-Lipschitz (and possibly also $\hat{\beta}-smooth$),
outputs a random point $\hat{\theta} \in K$ which is pure $\eps$-differentially private and satisfies $\mathbb{E}_{\hat{\theta}}[f(\hat{\theta},x)] - \min_{\theta \in K} f(\theta,x)\leq O(\frac{d \hat{L} R}{\epsilon})$.
Moreover, this algorithm takes at most   $T \times md^{\omega-1}$ arithmetic operations plus $T$ evaluations of the function $f$, where $T = O((md + d n^2 \eps^2) \times (\eps n + d  \mathrm{log}(\frac{nRd}{r\eps}))$ if each $\ell_i$ is $\hat{L}$-Lipschitz (or $T = O((md + d n \frac{\hat{\beta}}{\hat{L}} R \eps) \times (\eps n + d  \mathrm{log}(\frac{nRd}{r\eps}))$ if $f$ is also $\beta$-Lipschitz).
\end{corollary}
Corollary \ref{corr_DP} improves on \cite{bassily2014private,OM_Sampling22}. 
In particular, it improves upon the bound of  $O((\frac{1}{\epsilon^2}(m+n)d^{11}+ \epsilon^2 n^4 (m + n) d^7)  \times \mathrm{polylog}(\frac{nRd}{r\eps})))$ arithmetic operations in \cite{bassily2014private} by a factor of roughly $\max(\frac{d^{10-\omega}}{\epsilon^2 m},  \frac{1}{\epsilon}n d^{5})$, in the setting where the  $\ell_i$ are $\hat{L}$-Lipschitz on a polytope $K$ and each $\ell_i$ can be evaluated in $O(d)$ operations.
 And it improves by a factor of (at least) $md$ on the bound of $O((m^2d^3 + m^2d n^2 \eps^2) \times (\eps n + d) \mathrm{log}^2(\frac{nRd}{r\eps})) \times md^{\omega-1})$ arithmetic operations obtained in  \cite{OM_Sampling22}.
The proof is identical to that of how Theorem 2.2  implies Corollary 2.4 in \cite{OM_Sampling22}, and we refer the reader to \cite{OM_Sampling22}.
 For instance, when applying the exponential mechanism to the Lasso logistic regression problem, each loss $\ell_i$ is both $\hat{\beta}$-smooth and  $\hat{L}$-Lipschitz  for $\hat{\beta} = \hat{L} = 1$ and $R=O(1)$.
Thus, if $n=d$ and $\epsilon <1$, for this problem our algorithm requires $O(d^{3+\omega})$ arithmetic operations, which improves by a factor of $d^3$ on the bound of $d^{6+\omega}$ arithmetic operations implied by Corollary 2.4 in \cite{OM_Sampling22} and by roughly $d^{9-\omega}$ on the bound of $O(d^{12})$ arithmetic operations implied by \cite{bassily2014private}.
In another example, when training a support vector machine model with hinge loss and Lasso constraints under $\epsilon$-differential privacy, one has that $\ell_i$ is $\hat{L}$-Lipschitz but not smooth, for $\hat{L}=1$, and $R=O(1)$.
Thus, if $n=d$ and $\epsilon <1$, our algorithm requires $O(d^{4 +\omega})$ arithmetic operations,  which improves by a factor of $d^2$ on the bound of $d^{6+\omega}$ arithmetic operations implied by Corollary 2.4 in \cite{OM_Sampling22} and by roughly $d^{8-\omega}$ on the bound of $O(d^{12})$ implied by \cite{bassily2014private}.

\section{Overview of proof of Theorem \ref{thm_soft_threshold_Dikin} -- Main Result} \label{sec_technical_overview}
Suppose we are given any polytope $K=\{\theta \in \mathbb{R}^d : A \theta \leq b\}$ defined by $m$ inequalities, and a convex function $f: K \rightarrow \mathbb{R}^d$ which is  $L$-Lipschitz (or $\beta$-smooth) and given by an oracle which returns the value $f(\theta)$ at any point $\theta$.
Our  goal is to sample from the log-Lipschitz log-concave density $\pi \propto e^{-f}$ on $K$ within any total variation error $\delta>0$, 
in a number of arithmetic operations and oracle calls that has a dependence on the dimension $d$ that is a lower-order polynomial than currently available bounds for sampling from log-Lipschitz log-concave distributions, and is logarithmic in $\frac{1}{\delta}$.

\paragraph{Extending the Dikin walk to sample from log-concave distributions on polytopes.}  
As a first step, we begin by attempting to generate samples from $\pi$ via the (Gaussian) Dikin walk Markov chain, by extending the analysis given in \cite{sachdeva2016mixing} for the special case when $\pi$ is the uniform distribution on $K$ to the more general setting where $\pi$ is a log-Lipschitz log-concave density on a polytope $K$.

In the special case where $\pi$ is the uniform distribution on $K$, from any point $\theta$ in the interior of $K$, the Dikin walk proposes updates $z = \theta + \sqrt{\alpha H^{-1}(\theta)} \, \xi$ where $\xi \sim N(0,I_d)$ and $H(\theta) = \nabla^2 \varphi(\theta)$ is the Hessian of the log-barrier function $\varphi(\theta) = -  \sum_{j=1}^{m} \log(b_j-a_j^\top \theta )$ for $K=\{\theta \in \mathbb{R}^d : A \theta \leq b\}$, and $\alpha>0$ is a scalar hyperparameter.
To ensure that the stationary distribution of the Dikin walk is the uniform distribution on $K$, if a proposed update falls in the interior of $K$, it is accepted with probability $\min \left (\frac{\sqrt{\mathrm{det}(H(z))}}{\sqrt{\mathrm{det}(H(\theta))}} e^{\|z- \theta\|_{\Phi(\theta)}^2 - \|\theta- z\|_{\Phi(z)}^2}
 , 1 \right)$ determined by the metropolis rule; otherwise it is rejected.
The use of the log-barrier function is to ensure that the steps proposed by the Dikin walk Markov chain remain inside the polytope $K$ w.h.p.

The hyperparameter $\alpha^{-1}$ is chosen as large as possible while still ensuring that the proposed steps remain in $K$ and are accepted w.h.p.
On the one hand, since the covariance matrix  $\alpha^{-1} H^{-1}(\theta)$ of the proposed updates  is proportional to $\alpha^{-1}$, larger values of the hyperparameter $\alpha^{-1}$ allow the Dikin walk to propose larger update steps.
On the other hand, if one chooses $\alpha^{-1}$ too large, then the proposed steps of the Dikin walk may fall outside the polytope and be rejected with high probability, requiring the Dikin walk to propose a very large number of updates before it is able to take a step.
To  see how to choose the hyperparameter $\alpha^{-1}$, note that for any $\theta \in \mathrm{Int}(K)$, the Dikin ellipsoid $E(\theta) = \{w: w^\top H^{-1}(\theta) w \leq 1\}$ is contained in $K$.
Thus, standard Gaussian concentration inequalities which guarantee that $\| \xi\|_2 = O(\sqrt{d})$ w.h.p. whenever $\xi \sim N(0,I_d)$, imply that $\sqrt{\alpha^{-1} H^{-1}(\theta)} \, \xi$ is contained in $K$ w.h.p. if we set the step size parameter to be $\alpha^{-1} \leq O(\frac{1}{d})$. 
Moreover, using properties of log-barrier functions, one can show that the term $\frac{\mathrm{det}(H(z))}{\mathrm{det}(H(\theta))}$ in the acceptance ratio is also $\Omega(1)$  for $\alpha^{-1} = O(\frac{1}{d})$, as is done in \cite{kannan2012random, sachdeva2016mixing}.
To see why, Lemma 4.3 of \cite{vaidya1993technique}  implies that, if $H(\theta)$ is the Hessian of a log-barrier function for a polytope $K$,  then its log-determinant $V(\theta) = \log(\mathrm{det}(H(\theta)))$ satisfies the following inequality
\begin{equation} \label{eq_t4}
      (\nabla V(\theta)) ^\top [H(\theta)]^{-1} \nabla V(\theta) \leq O(d) \qquad \forall \theta \in \mathrm{Int}(K).
\end{equation}
Thus, if we choose $\alpha^{-1} \leq \frac{1}{d}$, the proposed update $z= \theta + \sqrt{\alpha^{-1} H^{-1}(\theta)} \, \xi$, where  $\xi \in N(0,I_d)$, has variance $\Omega(1)$ in the direction  $\nabla V(\theta)$, and (by standard Gaussian concentration inequalities), we have $(\theta-z)^\top \nabla V(\theta) \leq O(1)$ w.h.p.
This implies that $V(z) - V(\theta) = \log \frac{\mathrm{det}(H(z))}{\mathrm{det}(H(\theta))} = \Omega(1)$, and hence that $\frac{\mathrm{det}(H(z))}{\mathrm{det}(H(\theta))} = \Omega(1)$  w.h.p.

In \cite{narayanan2017efficient}, the Dikin walk Markov chain is applied to the more general problem of sampling from a $L$-Lipschitz (or $\beta$-smooth) log-concave distribution $\pi \propto e^{-f}$ on $K$ (the problem of interest in this paper). 
To guarantee that the Dikin walk has the correct stationary distribution $\pi$, the Metropolis acceptance probability of the proposed updates $z = \theta + \sqrt{\gamma^{-1} H^{-1}(\theta)} \,  \xi$, where $\gamma^{-1}$ is a hyperparameter,  gains an additional factor $\frac{e^{-f(z)}}{e^{-f(\theta)}}$. %
To ensure that this acceptance probability remains $\Omega(1)$, they modify the value of the scalar step size hyperparameter $\gamma^{-1}$ in such a way that w.h.p. the Markov chain takes steps where the value of $f$ changes by an amount at most $O(1)$.
To see how to choose $\gamma^{-1}$, note that since $f$ is $L$-Lipschitz, $e^{f(\theta)-f(z)}=\Omega(1)$ if the Euclidean distance $\|z-\theta\|_2$ is $O(\frac{1}{L})$.
This can be shown to occur with high probability if  $\gamma = O((\frac{1}{LR})^2)$, since the fact that the Dikin ellipsoid is contained in $K \subseteq B(0,R)$ implies that the eigenvalues of $H(\theta)$ must all be at most $R^2$ and hence that the variance of the proposed step would be at most $\frac{1}{dL^2}$ in any given direction.
Thus, it is sufficient for them to choose $\gamma = \min(\frac{1}{d}, (\frac{1}{LR})^2)$ to ensure the proposed step $z$ both remains in $K$ and is accepted with high probability by the Metropolis rule.

On the one hand, to ensure that the Markov chain proposes steps that change $f$ by an amount at most $O(1)$ for {\em any} $L$-Lipschitz function $f$, it is necessary and sufficient to ensure that from any point $\theta \in \mathrm{Int}(K)$, the Markov chain makes updates which fall w.h.p. inside a Euclidean ball $B(\theta, \frac{1}{L})$ of radius $\frac{1}{L}$ centered at $\theta$.
This is because the Lipschitz condition on $f$: $\|f(\theta) - f(z)\|_2 \leq L \|\theta - z \|_2$ for all $\theta, z \in K$ holds with respect to the Euclidean norm $\|\cdot \|_2$.
On the other hand, to ensure that the Markov chain remains inside the polytope $K$, it is sufficient for the Markov chain to propose steps which lie inside the Dikin ellipsoid $E(\theta) := \{w: w^\top H^{-1}(\theta) w \leq 1\}$ centered at $\theta$.
Roughly speaking, the scalar step size $\gamma^{-1}$ is chosen such that this Dikin ellipsoid is contained inside the Euclidean ball $B(\theta, \frac{1}{L})$, as this guarantees that w.h.p. the steps proposed by the Dikin walk will both remain inside the polytope $K$ and will also not change the value of $f$ by more than $O(1)$.

However, at many points $\theta$ the Dikin ellipsoid $E(\theta)$ is such that the ratio of the largest to smallest eigenvalues of $H^{-1}(\theta)$ may be very large (in fact this ratio can become arbitrarily large as $\theta$ approaches a face of the polytope).
Thus, roughly speaking, modifying the covariance matrix of the Dikin walk by a scalar constant ($\gamma^{-1}H^{-1}(\theta)$) can cause the Dikin walk to propose steps whose variance in some directions is much smaller than is required for {\em either} of the two goals: staying inside the polytope $K$ and staying inside the ball $B(\theta, \frac{1}{L})$ defined by the Lipschitz condition on $f$.
This suggests that modifying the log-barrier function for $K$ by a scalar multiple may not be the most efficient way of extending the Dikin walk to the problem of sampling from a general $L$-Lipschitz (or $\beta$-smooth) log-concave distribution on $K$, and 
that one may be able to obtain faster runtimes by making other modifications to the barrier function.

\paragraph{A soft-threshold version of the Dikin walk.} 

Before we introduce our soft-threshold Dikin walk, we first note that, even in the special case where $\pi$ is the uniform distribution on $K$, the analysis in \cite{narayanan2017efficient} does not recover the runtime bounds given in \cite{kannan2012random, sachdeva2016mixing} for this special case, as  \cite{narayanan2017efficient} use a different runtime analysis geared to time-varying distributions studied in that paper.
Namely, the results in \cite{narayanan2017efficient} imply a bound of $O(m^2 d^3 \log(\frac{\omega}{\delta}))$ Dikin walk steps to sample from a uniform distribution on $K$, while \cite{kannan2012random, sachdeva2016mixing} show bound of $O(md \log(\frac{\omega}{\delta}))$ in the special case where $\pi$ is the uniform distribution.
For this reason, we first extend the analysis of the Gaussian Dikin walk given in \cite{kannan2012random, sachdeva2016mixing} for the special case of uniform $\pi$, to the more general problem of sampling from an $L$-Lipschitz or $\beta$-mooth log-concave distribution. 
The analysis  in \cite{kannan2012random, sachdeva2016mixing}  uses the cross-ratio distance metric.
More specifically, if for any distinct points $u,v \in \mathrm{Int}(K)$ we let $p,q$ be the endpoints of the chord in $K$ which passes through $u$ and $v$ such that the four points lie in the order $p,u,v,q$, then cross-ratio distance is
\begin{equation}
\sigma(u,v) :=  \frac{\|u-v\|_2 \times \|p-q\|_2}{\|p-u\|_2 \times \|v-q\|_2}.
\end{equation}
As the usual Dikin walk Markov chain takes steps that have roughly identity covariance matrix $I_d$ with respect to the local norm  $\|u\|_{\gamma H(\theta)} := \sqrt{ u^\top \gamma  H(\theta) u}$,
and one can show that for any $u,v\in \mathrm{Int}(K)$, $\sigma^2(u,v) \geq \frac{1}{m \gamma^{-1}}\|u-v\|^2_{\gamma H(u)}$  (see e.g. \cite{kannan2012random, sachdeva2016mixing}), the variance in any given direction with respect to the cross-ratio distance is bounded below by  $\Delta = O(\frac{1}{m \gamma^{-1}})$.
Thus, for $\gamma = \min(\frac{1}{d}, (\frac{1}{LR})^2)$, the bound we would obtain on the number of steps until the Dikin walk is within TV error $\delta$ from $\pi$ is  $O(\Delta^{-1} \log(\frac{\omega}{\delta})) = O((md + mL^2R^2) \log(\frac{\omega}{\delta}))$ from an $\omega$-warm start.
To obtain even faster bounds, we would ideally like to allow the Dikin walk to take larger steps by choosing a larger value of $\gamma$, closer to the value of $\frac{1}{d}$ that is sufficient to ensure an $\Omega(1)$ acceptance probability in the special case when $\pi$ is uniform. 
Unfortunately, if e.g. $LR\geq d$, reducing $\gamma$  from a value of $\frac{1}{d}$ to a value of $(\frac{1}{LR})^2$, may be necessary to ensure that the variance of the Dikin walk steps is less than $\frac{1}{dL^2}$ in every direction, and hence that the acceptance probability is $\Omega(1)$.

To get around this problem, we introduce a new variant of the Dikin walk Markov chain for sampling from any $L$-Lipschitz (or $\beta$-smooth) log-concave distributions on a polytope $K$, which generalizes the Dikin walk Markov chain introduced in \cite{kannan2012random} for sampling from $\pi$ in the special case when $\pi$ is the uniform distribution on $K$.
The main difference between our Dikin walk and the usual Dikin walk of  \cite{kannan2012random} (and the version of the Dikin walk in \cite{narayanan2017efficient}) is that our Dikin walk regularizes the Hessian $\alpha^{-1}H(\theta)$ of the log-barrier for $K$ by adding a ``soft-thresholding'' term $\eta^{-1} I_d$ proportional to the identity matrix, where $\eta^{-1}$ is a hyperparameter and $\alpha^{-1}$ is the same hyperparameter appearing in the original Dikin walk of \cite{kannan2012random}.
Since the log-barrier Hessian $\alpha^{-1}H(\theta)$ and the regularization term $\eta^{-1} I_d$ have different scalar hyperparameters, we can set these two hyperparameters $\alpha^{-1}$ and $\eta^{-1}$ independently from each other: roughly speaking, $\alpha^{-1}$ is chosen to be the largest value such that the Dikin ellipsoid defined by the matrix $\alpha^{-1}H(\theta)$ remains inside the polytope $K$, while $\eta^{-1}$ is independently chosen to be the largest value such that, with high probability, the steps proposed by our Markov chain remain inside the ball $B(\theta, \frac{1}{L})$ defined by the Lipschitz condition on $f$.
Roughly speaking, this allows us to reduce the variance of the proposed steps of the Dikin walk Markov chain only in those directions where a choice of $\alpha = \frac{1}{d}$ would cause the variance to be greater than $\frac{1}{dL^2}$, while leaving the variance in other directions unchanged (up to a factor of 2).
More specifically, the steps proposed by our soft-threshold Dikin walk Markov chain are Gaussian with mean $0$ and covariance matrix $$\Phi^{-1}(\theta) := (\alpha^{-1} H(\theta) + \eta^{-1} I_d)^{-1},$$ for some hyperparameters $\alpha, \eta >0$.
The matrix $\Phi(\theta)$ is the Hessian of the function $\psi(\theta) = \alpha^{-1} \varphi(\theta) + \eta^{-1}\|\theta\|^2$ where $\varphi(\theta) = -  \sum_{j=1}^{m} \log(b_j-a_j^\top \theta )$ is the log-barrier for $K$.
The modified function $\psi(\theta)$ can be seen to also be a  self-concordant barrier function for the polytope $K$.
In the special case where $\pi$ is the uniform distribution, $L$ and  $\eta^{-1}$ are both equal to $0$, and our ``soft-threshold'' Dikin walk Markov chain recovers the original Dikin walk of \cite{kannan2012random}.
Thus, our  ``soft-threshold'' Dikin walk generalizes the original Dikin walk Markov chain to the problem of sampling from a general $L$-Lipschitz (or $\beta$-smooth) log-concave distribution on a polytope $K$.

\paragraph{Bounding the number of Markov chain steps.} 
Setting $\eta = \frac{1}{d L^2}$ ensures that the variance of the proposed update $z-\theta$ of our Markov chain is at most $O\left(\frac{1}{d L^2}\right)$ in any given direction, and hence that the term $e^{f(z)-f(\theta)}$ in the Metropolis acceptance rule is $\Omega(1)$ with high probability (Lemma \ref{lemma_density_ratio}).
Moreover, we also show that, if we choose $\alpha= \frac{1}{d}$, the other terms in the Metropolis acceptance rule are also $\Omega(1)$ (Lemmas \ref{lemma_remain_in_ellipsoid}, \ref{lemma_det}). 
While the proofs of these lemmas follow roughly the same outline as in the special case of the original Dikin walk where $\pi$ is uniform (e.g., \cite{kannan2012random, sachdeva2016mixing}), our bound on the determinantal term $\frac{\mathrm{det}\Phi(z)}{\mathrm{det}\Phi(\theta)}$ must deal with additional challenges, which we discuss in the next subsection.

To bound the number of steps required by our Markov chain to sample with TV error $O(\delta)$, we first bound the cross-ratio distance $\sigma(u,v)$ between any points $u,v \in \mathrm{Int}(K)$ by the local norm $\|u-v\|_{\Phi(u)}$ (Lemma \ref{lemma_cross_ratio}):
\begin{align} \label{eq_t1}
     \sigma^2(u,v) 
        \geq \left(\frac{1}{2m} \sum_{i=1}^m \frac{(a_i^\top(u-v))^2}{(a_i^\top u-b_i)^2} \right) + \frac{1}{2} \frac{\|u-v\|_2^2}{R^2}
                \geq \frac{1}{2m \alpha^{-1} + 2 \eta^{-1} R^2}\|u-v\|^2_{\Phi(u)}.
\end{align}
Roughly, this means the variance of the cross-ratio distance of our Markov chain's step in any given direction is bounded below by some number $\Delta$, where $\Delta = \Omega\left(\frac{1}{2m \alpha^{-1} + 2 \eta^{-1} R^2}\right)$.
Using \eqref{eq_t1} together with the isoperimetric inequality for the cross-ratio distance  (Theorem 2.2 of \cite{lovasz2003hit}), we show that, if the acceptance probability of our Markov chain is $\Omega(1)$ at each step, then the number of steps for our Markov chain to obtain a sample within a TV distance of $\delta$ from $\pi$ is $O(\Delta^{-1} \log(\frac{\omega}{\delta})) %
= O((md + d L^2 R^2)\log(\frac{\omega}{\delta}))$ from an $\omega$-warm start.
In particular, in the regime where $m=O(d)$ and $LR > d$, this improves on the bound we would get for the basic Dikin walk by a factor of $d$.

\paragraph{Bounding the determinantal term in the acceptance probability.} 
For our mixing time bound to hold, we still need to show that the determinantal term $\frac{\mathrm{det}(\Phi(z))}{\mathrm{det}(\Phi(\theta))}$ is $\Omega(1)$ with high probability.
To bound this term, we would ideally like to follow the general approach that previous works \cite{kannan2012random, sachdeva2016mixing} use to show that the determinantal term $\frac{\mathrm{det}(H(z))}{\mathrm{det}(H(\theta))}$ in the  basic Dikin walk is $\Omega(1)$ with high probability, which relies on the  property of log-barrier functions in Inequality \eqref{eq_t4}.
Unfortunately, since $\Phi(\theta)$ is not the Hessian of a log-barrier function for any system of inequalities defining the polytope $K$, we cannot directly apply Inequality \eqref{eq_t4} to $\Phi(\theta)$. 

To get around this problem, we show that, while $\Phi(\theta)$ is not the Hessian of a log-barrier function, it is in fact the limit of a sequence of matrices $H_i(\theta)$, $i\in \mathbb{N}$, where each matrix $H_i(\theta)$ in this sequence {\em is} the Hessian of a (different) log-barrier function for $K$. 
Specifically, for every $j \in \mathbb{N}$, we consider the matrices $A^j = [A^\top, I_d, \ldots, I_d]^\top$ where $A$ is concatenated with $\frac{m_j -1}{d}$ copies of the identity matrix $I_d$,  and $m_j  = m + \lfloor \alpha \eta^{-1} j^2 \rfloor d$.
And we consider the vectors $b^j = (b^\top, j \textbf{1}^\top, \ldots, j \textbf{1}^\top)^\top$, where $b$ is concatenated with $\frac{m_j -1}{d}$ copies of the vector $j \textbf{1}$, where $\textbf{1}= (1,\ldots,1)^\top \in \mathbb{R}^d$ is the all-ones vector.
Then, for large enough $j$, we have $K = \{\theta \in \mathbb{R}^d: A^j \theta \leq b^j\}$, and the Hessian of the corresponding log-barrier functions is
\begin{align} \label{eq_t8}
  H_j(w) &=\sum_{i=1}^{m^j} \frac{a^j_i (a^j_i)^\top}{((a^j_i)^\top w -b^j_i)^2}
=H(w) +  \lfloor \alpha \eta^{-1} j^2 \rfloor  \sum_{i=1}^{d} \frac{e_i e_i^\top}{(e_i^\top w -j)^2}.
\end{align}
 Using this fact \eqref{eq_t8}, we show that, for every $\theta \in \mathrm{int}(K)$, every $z \in \frac{1}{2}D_\theta$, and every sequence $\{z_j\}_{j=1}^\infty \subseteq \frac{1}{2}D_\theta$ such that $z_j \rightarrow z$,  we have (Lemma \ref{lemma_limits}),
 \begin{equation} \label{eq_t7}
  \lim_{j\rightarrow \infty}     \frac{\mathrm{det}(H_j(z_j))}{\mathrm{det}(H_j(\theta))} = \frac{\mathrm{det}(\Phi(z))}{\mathrm{det}(\Phi(\theta))}.
\end{equation}
Moreover, since each $H_j$ is the Hessian of a log-barrier function for $K$, we have that \eqref{eq_t4} does hold for $H_j$ and hence (from the work of \cite{kannan2012random, sachdeva2016mixing}) that %
$\frac{\mathrm{det}(H_j(z_j))}{\mathrm{det}(H_j(\theta))} = \Omega(1)$
  w.h.p. for all $j \in \mathbb{N}$, if we set 
$z_j = \theta+ \alpha^{\frac{1}{2}} H_j^{-\frac{1}{2}}(\theta) \xi$, $\xi \sim N(0, I_d)$, and choose $\alpha = \frac{1}{d}$.
Thus,  \eqref{eq_t7} implies that $\frac{\mathrm{det}(\Phi(z))}{\mathrm{det}(\Phi(\theta))} = \Omega(1)$ w.h.p. as well (Lemma \ref{lemma_det}), and hence the acceptance probability of the soft-threshold Dikin walk is $\Omega(1)$ at each step.

\paragraph{Bounding the number of arithmetic operations.} \label{sec_soft_threshold_concluding}

Since the acceptance probability is $\Omega(1)$, from the above discussion we have that the number of steps for our Markov chain to obtain a sample within a TV distance $\delta>0$ from $\pi$ is $O((md + L^2 R^2)\log(\frac{\omega}{\delta}))$ from an $\omega$-warm start.

Each time our Markov chain proposes a step $\theta + \Phi(\theta)^{-\frac{1}{2}} \xi$, it must first sample a Gaussian vector $\xi \sim N(0,I_d)$ which takes $O(d)$ arithmetic operations.
 To compute $\Phi(\theta)$, it must then compute the log-barrier Hessian $H(\theta)$, and invert the matrix $\Phi(\theta) = \alpha^{-1} H(\theta) + \eta^{-1} I_d$.

 Since   $H(w) = C(w) C(w)^\top$, where $C(w)$ is a $d \times m$ matrix with columns $c_j(w) = \frac{a_j}{a_j^\top w - b_j}$ for all $j \in [m]$, we can compute $H(w)$ in $m d^{\omega -1}$ arithmetic operations using fast matrix multiplication.
    And since $\Phi(\theta)$ is a $d \times d$ matrix, computing $\Phi(\theta)^{-\frac{1}{2}}$ can be accomplished in $d^\omega$ arithmetic operations by computing the singular value decomposition of $\Phi$.
    Next, the Markov chain must compute the acceptance probability\\ $$  \min  \left (\frac{e^{-f(z)} \sqrt{\mathrm{det}(\Phi(z))}}{e^{-f(\theta)} \sqrt{\mathrm{det}(\Phi(\theta))}} e^{\|z- \theta\|_{\Phi(\theta)}^2 - \|\theta- z\|_{\Phi(z)}^2}
 , 1 \right).$$
    Here, the two determinantal terms can be computed in $O(d^\omega)$ arithmetic operations by again computing the singular value decomposition of $\Phi$,
    and the two values of $f$ can be computed in two calls to the oracle for $f$.

    Thus, from an $\omega$-warm start, the soft-threshold Dikin walk takes at most $O((md + d L^2 R^2) \times \log(\frac{w}{\delta}))$ Markov chain steps to obtain a sample from $\pi$ with total variation error $\delta>0$, where each step takes $O(md^{\omega-1})$ arithmetic operations, and one function evaluation.

\section{Extension to general barrier functions?}\label{sec:barrier}

From any point $\theta$, our algorithm proposes a step with Gaussian distribution $N(\theta, \frac{1}{d}(\nabla^2 g(\theta))^{-1})$, where $g$ is the following barrier function
\begin{equation} \label{eq_soft_threshold_barrier}
g(\theta) = \varphi(\theta) + \hat{\eta}^{-1} \theta^\top \theta,
\end{equation}
where $\varphi(\theta)$ is a barrier function for the polytope $K$, and the parameter $\hat{\eta}^{-1} = \Omega(\frac{1}{L^2})$ in the setting where $f$ is guaranteed to be $L$-Lipschitz and $\hat{\eta}^{-1} = \Omega(\frac{1}{\beta})$ in the setting where $f$ is guaranteed to be $\beta$-smooth.
Thus, most of the probability mass of the Gaussian distribution concentrates in the Dikin ellipsoid $E(\theta) := \theta + \{w: w^\top (\nabla^2 g(\theta))^{-1} w \leq 1\}$ for the barrier function $g$.
For simplicity, in our main result, we assume that $\varphi(\theta):= -  \sum_{j=1}^{m} \log(b_j-a_j^\top \theta )$ (which has self-concordance parameter $m$), however, we can in principle choose any barrier function $\varphi$ for the polytope $K$, such as the entropic barrier \cite{bubeck2015entropic} or the Lee-Sidford Barrier \cite{lee2019solving} which have self-concordance parameter roughly $\nu = d$.

To arrive at our barrier function from a more axiomatic approach, we first consider the definition of $\nu$-self concordant barrier function:
\begin{definition} [$\nu$-self-concordant barrier function for $K$] \label{def_barrier}
We say that $g$ is a $\nu$-self-concordant barrier function if $g: \mathrm{Int}(K) \rightarrow \mathbb{R}$ and $g$ satisfies the following conditions:
\begin{enumerate}
    \item \textbf{Convex and differentiable barrier:} $g$ is convex and third-order differentiable, and $g(x) \rightarrow + \infty$ as $x \rightarrow \partial K$.

\item \textbf{Self-concordance:}   $\nabla^3g(x)[h,h,h] \leq 2 (\nabla^2g(x)[h,h])^{3/2}$ for all $h \in \mathbb{R}^d$            (this ensures that the Hessian of the barrier function does not change too much each time the Dikin walk takes a step from $x$ to $z$,  that is, $\frac{1}{2} \nabla^2g(z) \preceq \nabla^2g(x) \preceq 2 \nabla^2g(z)$)

\item  \textbf{$g$ is $\nu$-self concordant:} 
$h^\top \nabla g(x) \leq \sqrt{\nu h^\top \nabla^2 g(x) h}$ for every $x \in \mathrm{Int}(K)$,  $h \in \mathbb{R}^d$
\end{enumerate}
\end{definition}
The fact that our barrier function \eqref{eq_soft_threshold_barrier} satisfies parts (1) and (2) of Definition \ref{def_barrier} follows from the fact that $\varphi$ satisfies Definition \ref{def_barrier} and that $\nabla^3(\theta^\top \theta) = 0$.  We discuss the self-concordance parameter $\nu$ for which our barrier function satisfies  \eqref{eq_soft_threshold_barrier}  below.

 To ensure that the steps $z \sim N(\theta, (\frac{1}{d}\nabla^2 g(\theta))^{-1})$ proposed by our Dikin walk Markov chain arising from the barrier function $g$ has an $\Omega(1)$ acceptance ratio $e^{f(z)-f(\theta)}$, we require that the function $g$ satisfies the following property.
This property says that at least $\frac{1}{4}$ of the volume of the Dikin ellipsoid $E(\theta)$ is contained in the sublevel set $\{z \in K: f(\theta) < f(\theta) + 2 \}$ where the value of $f$ does not increase by more than $2$.
\begin{property}\label{property_ellipsoid}\textbf{(Dikin ellipsoid  mostly contained in sublevel set)}  At every $\theta \in \mathrm{Int}(K)$, the Dikin Ellipsoid, satisfies
\begin{eqnarray*}
  \mathrm{Vol}\left( E(\theta) \cap \{z \in K: f(z) < f(\theta) + 2 \} \right) \geq \frac{1}{4} \mathrm{Vol}(E(\theta)).
\end{eqnarray*}
\end{property}
When designing a barrier function $g$, there is a trade-off between choosing $g$ such that the self-concordance parameter $\nu$ is small, while at the same time ensuring that  Property \ref{property_ellipsoid} holds.
Roughly, to make the parameter $\nu$ as small as possible, we would like the Dikin ellipsoid to be large relative to the Hilbert-distance metric for the convex body $K$ (This is because by Proposition 2.3.2(iii) of \cite{nesterov1994interior}, any $\nu$-self concordant function $g$ satisfies $(h^\top \nabla^2 g(x) h)^{-\frac{1}{2}} \leq |h|_x \leq (1+ 3\nu) (h^\top \nabla^2 g(x) h)^{-\frac{1}{2}}$, for any $h \in \mathbb{R}^d$ where $|h|_x := \sup \{ \alpha>0 :   x \pm \alpha h \in  \{z \in K \}$).
On the other hand, if we make the Dikin ellipsoid too large (with respect to the sublevel set $\{z \in K: f(z) < f(x) + 2 \}$) then  Property \ref{property_ellipsoid} will not be satisfied, and the steps proposed by the Dikin walk will have a very low acceptance probability.
Setting the hyperparameter $\hat{\eta} = \frac{1}{L^2}$ when $f$ is $L$-Lipschitz or $\hat{\eta} = \frac{1}{\beta}$ when $f$ is $\beta$ smooth ensures that our barrier function in  \eqref{eq_soft_threshold_barrier} satisfies  Property \ref{property_ellipsoid}.

In the special case where $f$ is the log-density of the uniform distribution ($f\equiv 0$), or when $f$ is any linear function, we have that $f$ is $\beta$-smooth for $\beta=0$.
Thus, our barrier function \eqref{eq_soft_threshold_barrier} we use to encode the geometry of $f$ on $K$ is the same as the barrier function $\varphi$ for the polytope $K$.
This is because, since the level sets of linear functions $f$ are half-planes,  {\em any} ellipsoid centered at $\theta$ has at least half of its volume in the sublevel set $\{z \in K: f(z) \leq f(\theta) \} \subseteq \{z \in K: f(z) < f(\theta) + 2 \}$, satisfying Property \ref{property_ellipsoid}.

The following lemma shows that our barrier function in \eqref{eq_soft_threshold_barrier} is $\nu$-self concordant with $\nu = O(\nu' + \hat{\eta} R^2)$, where $\nu'$ is the self-concordance parameters of $\varphi$ (which is $\nu'= m$ if we choose $\varphi$ to be log-barrier function).
Thus, our barrier function in \eqref{eq_soft_threshold_barrier} is $\nu = O(\nu' + L^2R^2)$ self-concordant in the setting where $f$ is $L$-Lipschitz,  and $\nu = O(\nu' + \beta R^2)$ in the setting where $f$ is $\beta$-smooth:

\begin{lemma}\label{lemma_nu}
Suppose that $\phi(x)$ is a $\nu'$-self concordant barrier function for a convex body $K \subset \mathbb{R}^d$ where $B(0,r) \subseteq K \subseteq B(0,R)$ for some $R>r>0$.
Let $g(x) = \phi(x) + \frac{\alpha}{2} x^\top x$ for some $\alpha>0$.
Then $g$ is $\nu$-self concordant for $\nu = 4 \nu' + 4\alpha R^2$.
\end{lemma}
The proof of Lemma \ref{lemma_nu} is given in Appendix \ref{sec_lemma_nu}. 
The polynomial dependence of $\nu$ on $LR$ or $\beta R^2$ is a necessary feature of any objective function satisfying Definition \ref{def_barrier} and Property \ref{property_ellipsoid}.
In appendix \ref{appendix_lower_bounds_nu}, we give explicit examples of classes of objective functions $f$ and polytopes $K$ for which the minimum value of $\nu$ depends polynomially on $LR$ (and classes of smooth functions $f$ where $\nu$ depends polynomially on $\sqrt{\beta}R$.

An open problem is whether one can design versions of the Dikin walk which sample from log-concave distributions with runtime that depends only on the dimension $d$, and is independent of $L$, $R$ or $\beta$ and which are invariant to affine transformations.
The difficulty in achieving this lies in the fact that (e.g., in the setting where $f$ is $L$-Lipschitz), on the one hand, the level sets of $f$ where most of the probability mass of $\propto e^{-f}$ concentrates may have a diameter roughly $RL$ times smaller than the diameter $R$ of $K$.
Thus, to have an $\Omega(1)$ acceptance probability the Dikin walk may need to take steps that  are roughly $RL$ times smaller than the diameter of $K$.
On the other hand, the isoperimetric inequality currently used to bound the mixing time of the Dikin walk uses a metric-- the cross-ratio distance for $K$--which, roughly speaking, defines distances between steps by how quickly these steps approach the boundary of $K$.
Thus, measured in the cross-ratio distance, the steps proposed by the Dikin walk are of size roughly proportional to $\frac{1}{RL}$, and mixing time bounds obtained with this isoperimetric inequality thus depend polynomially on $RL$.
To obtain mixing time bounds independent of the $R, L, \beta$ one may need to show a new isoperimetric inequality which is based on a different metric that encodes the geometry of all of the level sets of $f$--rather than just the geometry of its support $K$.

\section{Algorithm}\label{sec_Algorithms}

\begin{algorithm}[h]
\caption{Soft-threshold Dikin walk} \label{alg_Soft_Dikin_Walk}
\KwIn{$m,d \in \mathbb{N}$}
\KwIn{$A \in \mathbb{R}^{m \times d}$, $b \in \mathbb{R}^m$, which define the polytope $K := \{\theta \in \mathbb{R}^d: A \theta \leq b\}.$}

\KwIn{Oracle which returns the value of a convex function $f: K \rightarrow \mathbb{R}$.}
\KwIn{An initial point $\theta_0 \in \mathrm{Int}(K)$.}
\medskip
 \KwOut{A point $\theta$.}

\medskip
  
\KwHyperparameters{ $\alpha>0$, $\eta>0$, and $T\in \mathbb{N}$.}

\medskip

Set $\theta \leftarrow \theta_0$

\For{$i = 1, \ldots, T$}{

Sample a point $\xi \sim N(0,I_d)$ \label{Line_Gaussian}

Set $H(\theta) \leftarrow \sum_{j=1}^{m} \frac{a_j a_j^\top}{(a_j^\top \theta - b_j)^2}$ \label{line_Hessian_1}

Set $\Phi(\theta) \leftarrow \alpha^{-1} H(\theta) +  \eta^{-1}I_d$ \label{line_Phi_1}

Set $z \leftarrow \theta + \Phi(\theta)^{-\frac{1}{2}} \xi$ \label{line_proposed_update}

\If{$z \in \mathrm{Int}(K)$}{\label{Line_membership}

Set $H(z) \leftarrow \sum_{j=1}^{m} \frac{a_j a_j^\top}{(a_j^\top z - b_j)^2}$  \label{line_Hessian_2}

Set $\Phi(z) \leftarrow \alpha^{-1} H(z) +  \eta^{-1}I_d$ \label{line_Phi_2}

Accept $\theta \leftarrow z$ with probability $\frac{1}{2} \times \min \left (\frac{e^{-f(z)} \sqrt{\mathrm{det}(\Phi(z))}}
{e^{-f(\theta)} \sqrt{\mathrm{det}(\Phi(\theta))}} \times e^{\|z- \theta\|_{\Phi(\theta)}^2 - \|\theta- z\|_{\Phi(z)}^2}, \, \, 1 \right) $
} \label{Line_accept_reject}

\Else{Reject $z$}
}
Output $\theta$

\end{algorithm}

\noindent 
In Theorem \ref{thm_soft_threshold_Dikin}, we set the step size hyperparameters 
$$\alpha = \frac{1}{10^5 d} \mbox{  and  } \eta = \frac{1}{10^4 d  L^2}$$ if $f$ is $L$-Lipschitz, and the number of steps to be $$T= 10^9 \left( 2m \alpha^{-1} + \eta^{-1} R^{2} \right) \times \log(\frac{w}{\delta}).$$
  To obtain the bounds 
  when $f$ is $\beta$-smooth (but not necessarily Lipschitz), we instead set 
  $$\alpha = \frac{1}{10^5 d} \mbox{  and  } \eta = \frac{1}{10^4 d  \beta}.$$

\section{Proof of Theorem \ref{thm_soft_threshold_Dikin}}\label{sec_mainproof}

\subsection{Bounding the number of arithmetic operations}

In the following, we assume the hyperparameters $\alpha, \eta$ satisfy $\alpha \leq \frac{1}{10^5 d}$, and either $\eta \leq \frac{1}{10^4 d  L^2}$ (in the setting where $f$ is $L$-Lipschitz) or $\eta \leq \frac{1}{10^4 d  \beta}$ (in the setting where $f$ is $\beta$-smooth).

\begin{lemma} \label{Lemma_operation_count}
Each iteration of Algorithm \ref{alg_Soft_Dikin_Walk} can be implemented in $ O(md^{\omega-1})$ arithmetic operations plus $O(1)$ calls to the oracle for the value of $f$. 
\end{lemma}

\begin{proof}
We go through each step of Algorithm \ref{alg_Soft_Dikin_Walk} and add up the number of arithmetic operations and oracle calls for each step:

\begin{enumerate}
 \item Line \ref{Line_Gaussian} samples a $d$-dimensional Gaussian random vector $\xi \sim N(0, I_d)$, which can be performed in $O(d)$ arithmetic operations.
 
    \item 
At each iteration, Algorithm \ref{alg_Soft_Dikin_Walk} computes the matrix $H(w)$ at $w=\theta$ (line \ref{line_Hessian_1}) and $w=z$ (line \ref{line_Hessian_2}):
\begin{align*}
    H(w) = \sum_{j=1}^{m} \frac{a_j a_j^\top}{(a_j^\top w - b_j)^2}.
\end{align*}
Computing $H(w)$ at any $w\in \mathbb{R}^d$ can be accomplished in $m d^{\omega -1}$ operations as follows:

Define $C(w)$ to be the $d \times  m$ matrix where each column $c_j(w) = \frac{a_j}{a_j^\top w - b_j}$ for all $j \in [m]$.
Then
\begin{equation*}
    H(w) = C(w) C(w)^\top.
\end{equation*} 
Since $C(w)$ is a $d \times m$ matrix the product $C(w) C(w)^\top$ can be computed in $m d^{\omega -1}$ arithmetic operations.
Thus, Lines \ref{line_Hessian_1} and \ref{line_Hessian_2} of  Algorithm \ref{alg_Soft_Dikin_Walk} can each be computed in $m d^{\omega -1}$ arithmetic operations.

\item Since Lines \ref{line_Phi_1} and \ref{line_Phi_2}  compute a sum of two $d \times d$ matrices, Lines \ref{line_Phi_1} and \ref{line_Phi_2} can each be performed in $d^2$ arithmetic operations.

\item 
Line \ref{line_proposed_update} computes the proposed update
\begin{equation*}
    z = \theta + \Phi(\theta)^{-\frac{1}{2}} \xi.
\end{equation*}
    Computing $\Phi(\theta)^{-\frac{1}{2}}$ can be performed by taking the singular value decomposition of the matrix $\Phi(\theta)$, and then inverting and taking the square root of its eigenvalues.
    This can be accomplished in $d^\omega$ arithmetic operations since $\Phi(\theta)$ is a $d \times d$ matrix.
    Once  $\Phi(\theta)^{-\frac{1}{2}}$ is computed, the computation $\theta +  \Phi(\theta)^{-\frac{1}{2}} \xi$ can be performed in $d^2$ arithmetic operations.
    Thus Line \ref{line_proposed_update} can be computed in $O(d^\omega) \leq O(m d^{\omega-1})$ arithmetic operations.
    
    \item Line \ref{Line_membership} requires determining whether $z \in K$. 
    This can be accomplished in $O(md)$ arithmetic operations, by checking whether the inequality $Az \leq b$ is satisfied.
    
    \item Line \ref{Line_accept_reject}  requires computing the determinant $\mathrm{det}(\Phi(w))$ and $f(w)$ for $w = \theta$ and $w = z$. 
    Computing $\mathrm{det}(\Phi(w))$ can be accomplished by computing the singular value decomposition of $\mathrm{det}(\Phi(w))$ and then taking the product of the resulting singular values to compute the determinant.
    Since $\Phi(w)$ is a $d \times d$ matrix, computing the singular value decomposition can be done in $d^\omega$ arithmetic operations.
    Computing $f(w)$ for any $w \in \mathbb{R}^d$ can be accomplished in one call to the oracle for the value of $f$.
    Thus, Line \ref{Line_accept_reject} can be computed in $O(d^\omega) \leq O(m d ^{\omega -1})$ arithmetic operations, plus two calls to the oracle for the value of $f$.
\end{enumerate}
Therefore, adding up the number of arithmetic operations and oracle calls from all the different steps of Algorithm \ref{alg_Soft_Dikin_Walk}, we get that each iteration of Algorithm \ref{alg_Soft_Dikin_Walk} can be computed in $ O(md^{\omega-1})$ arithmetic operations plus $O(1)$ calls to the oracle for the value of $f$.

\end{proof}

\subsection{Bounding the step size}

\begin{definition}[Cross-ratio distance]

Let $u,v \in \mathrm{Int}(K)$.  If $u \neq v$, let $p,q$ be the endpoints of the chord in $K$ which passes through $u$ and $v$ such that the four points lie in the order $p,u,v,q$. 
Define 
\begin{equation*}
\sigma(u,v) :=  \frac{\|u-v\|_2 \times \|p-q\|_2}{\|p-u\|_2 \times \|v-q\|_2}
\end{equation*}
if $u \neq v$, and $\sigma(u,v) = 0$ if $u=v$.
\end{definition}
For convenience, we define the cross-ratio distance between any two subsets $S_1, S_2 \subseteq K$ as
\begin{equation*}
    \sigma(S_1, S_2) = \min\{\sigma(u, v) : u \in S_1, v \in S_2\}.
\end{equation*}
And for any $S \subseteq \mathbb{R}^d$ and any density $\nu: \mathbb{R}^d \rightarrow \mathbb{R}$ we define the induced measure:
\begin{equation*}
    \nu^\star(S) = \int_{z \in S} \nu(z) \mathrm{d}z.
\end{equation*}

\begin{definition}
For any positive-definite matrix $M \in \mathbb{R}^d \times \mathbb{R}^d$, and any $u \in \mathbb{R}^d$, we define
\begin{equation*}
     \|u\|_{M} := \sqrt{ u^\top Mu}.
\end{equation*}
\end{definition}

\begin{lemma}\label{lemma_cross_ratio}
For any $u,v \in \mathrm{Int}(K)$, we have
\begin{equation*}
    \sigma(u,v) \geq \frac{1}{\sqrt{2m \alpha^{-1} + \eta^{-1} R^{2}}} \|u-v\|_{\Phi(u)}.
\end{equation*}
\end{lemma}

\begin{proof}
Let $p,q$ be the endpoints of the chord in $K$ which passes through $u$ and $v$ such that the four points lie in the order $p,u,v,q$.
Then
\begin{align*}
    \sigma^2(u,v) &= \left(\frac{\|u-v\|_2 \times \|p-q\|_2}{\|p-u\|_2 \times \|v-q\|_2}\right)^2\\
    &\geq \max \left(\frac{\|u-v\|_2^2}{\|p-u\|_2^2}, \, \, \frac{\|u-v\|_2^2}{\|u-q\|_2^2}, \, \, \frac{\|u-v\|_2^2}{\|p-q\|_2^2}\right)\\
    &= \max \left( \max_{i \in [m]} \frac{(a_i^\top(u-v))^2}{(a_i^\top u-b_i)^2}, \, \, \,  \frac{\|u-v\|_2^2}{\|p-q\|_2^2} \right)\\
    &\geq \frac{1}{2} \max_{i\in [m]} \frac{(a_i^\top(u-v))^2}{(a_i^\top u-b_i)^2} + \frac{1}{2} \frac{\|u-v\|_2^2}{\|p-q\|_2^2}\\
        &\geq \left(\frac{1}{2m} \sum_{i=1}^m \frac{(a_i^\top(u-v))^2}{(a_i^\top u-b_i)^2} \right) + \frac{1}{2} \frac{\|u-v\|_2^2}{R^2}\\
                &= (u-v)^\top \left(\frac{1}{2m \alpha^{-1}} \times \alpha^{-1} \sum_{i=1}^m \frac{(a_i a_i^\top)^2}{(a_i^\top u-b_i)^2} + \frac{1}{2 R^2 \eta^{-1}} \times \eta^{-1} I_d   \right) (u-v) \\
                &\geq \frac{1}{2m \alpha^{-1} + 2 \eta^{-1} R^2} (u-v)^{\top} \Phi(u) (u-v)\\
                &
                = \frac{1}{2m \alpha^{-1} + 2 \eta^{-1} R^2}\|u-v\|_{\Phi(u)}^2.
\end{align*}
Thus, we have
\begin{equation*}
    \sigma(u,v) \geq \frac{1}{\sqrt{2m \alpha^{-1} + \eta^{-1} R^{2}}} \|u-v\|_{\Phi(u)}.
\end{equation*}
\end{proof}

\begin{lemma} \label{lemma_PD}
For any $u,v \in \mathrm{Int}(K)$ such that $\|u-v\|_{\Phi(u)} \leq \frac{1}{2\alpha^{\nicefrac{1}{2}}}$ we have that
\begin{equation*}
(1- \alpha^{\nicefrac{1}{2}} \|u-v\|_{\Phi(u)})^2 \Phi(v) \preceq \Phi(u)  \preceq (1 + \alpha^{\nicefrac{1}{2}} \|u-v\|_{\Phi(u)})^2 \Phi(v).
\end{equation*}
\end{lemma}

\begin{proof}
\begin{align*}
    \|u-v\|_{\Phi(u)}^2 &= \alpha^{-1} \sum_{i=1}^m \frac{(a_i^\top(u-v))^2}{(a_i^\top u-b_i)^2}  + \eta^{-1}\|u-v\|^2\\
    &\geq \alpha^{-1} \max_{i \in [m]} \frac{(a_i^\top(u-v))^2}{(a_i^\top u-b_i)^2}.
\end{align*}
Thus,
\begin{equation*}
    |(a_i^\top u - b_i) - (a_i^\top v - b_i)| \leq  \alpha^{\nicefrac{1}{2}}  \|u-v\|_{\Phi(u)} |a_i^\top u - b_i| \qquad \forall i \in [m].
\end{equation*}
Therefore, for all $w \in \mathbb{R}^d$ we have
\begin{align*}
    &w^\top \left[(1- \alpha^{\nicefrac{1}{2}}  \|u-v\|_{\Phi(u)})^2 \alpha^{-1} \sum_{i=1}^m \frac{a_i a_i^\top}{(a_i^\top v-b_i)^2} + \eta^{-1} I_d \right] w\\
    & \leq  w^\top \left[ \alpha^{-1} \sum_{i=1}^m \frac{a_i a_i^\top}{(a_i^\top u-b_i)^2} + \eta^{-1} I_d \right] w\\
    &\leq w^\top \left[(1+ \alpha^{\nicefrac{1}{2}}  \|u-v\|_{\Phi(u)})^2 \alpha^{-1} \sum_{i=1}^m \frac{a_i a_i^\top}{(a_i^\top v-b_i)^2} + \eta^{-1} I_d \right] w.
\end{align*}
Thus,

\begin{equation*}
    (1- \alpha^{\nicefrac{1}{2}}  \|u-v\|_{\Phi(u)})^2 \Phi(v) \preceq \Phi(u) \preceq (1+ \alpha^{\nicefrac{1}{2}}  \|u-v\|_{\Phi(u)})^2 \Phi(v).
\end{equation*}
\end{proof}

\subsection{Bounding the acceptance probability}

\begin{lemma} \label{lemma_acceptance_ratio}
Let $\theta \in \mathrm{Int}(K)$. Then the acceptance ratio  satisfies
\begin{equation*}
\mathbb{P}_{z \sim N(\theta, \Phi^{-1}(\theta))} \left( \frac{\pi(z) \sqrt{\mathrm{det}(\Phi(z))}}
{\pi(\theta) \sqrt{\mathrm{det}(\Phi(\theta))}} \times \exp\left(\|z- \theta\|_{\Phi(\theta)}^2 - \|\theta- z\|_{\Phi(z)}^2\right) \times \mathbbm{1}\{z\in K\} \geq \frac{3}{10} \right) \geq \frac{1}{3}.
\end{equation*}
\end{lemma}
\begin{proof}
By \eqref{eq_f13} of Lemma \ref{lemma_det}, we have
\begin{equation}\label{eq_f13b}
\mathbb{P}_{z \sim N(\theta, \Phi^{-1}(\theta))}\left( \|z - \theta\|_{\Phi(\theta)}^2 - \|z - \theta\|_{\Phi(z)}^2
\geq -\frac{2}{50} \right )\geq \frac{98}{100}.
\end{equation}
By  Lemmas \ref{lemma_density_ratio}, \ref{lemma_det}, and \ref{lemma_remain_in_ellipsoid} we have that
\begin{align*}
&\mathbb{P}_{z \sim N(\theta, \Phi^{-1}(\theta))} \left( \frac{\pi(z) \sqrt{\mathrm{det}(\Phi(z))}}{\pi(\theta) \sqrt{\mathrm{det}(\Phi(\theta))}} \times \exp\left(\|z- \theta\|_{\Phi(\theta)}^2 - \|\theta- z\|_{\Phi(z)}^2\right)  \times \mathbbm{1}\{z\in K\} \geq \frac{3}{10} \right)\\
&\geq \mathbb{P}_{z \sim N(\theta, \Phi^{-1}(\theta))} \bigg( \left\{ \frac{\pi(z)}{\pi(\theta)} \geq \frac{1}{2} \right \} \cap \left\{ \frac{\sqrt{\mathrm{det}(\Phi(z))}}{ \sqrt{\mathrm{det}(\Phi(\theta))}} \geq \frac{48}{50} \right \}\\
&\qquad \qquad \qquad \qquad \qquad \qquad \cap \left\{\exp\left(\|z- \theta\|_{\Phi(\theta)}^2 - \|\theta- z\|_{\Phi(z)}^2\right) \geq 0.96\right\} \cap \{ z \in K \} \bigg)\\
&= 1 - \mathbb{P}_{z \sim N(\theta, \Phi^{-1}(\theta))} \bigg( \left\{ \frac{\pi(z)}{\pi(\theta)} \geq \frac{1}{2} \right \}^c \cup \left\{ \frac{\sqrt{\mathrm{det}(\Phi(z))}}{ \sqrt{\mathrm{det}(\Phi(\theta))}} \geq \frac{48}{50} \right \}^c\\
&\qquad \qquad \qquad \qquad \qquad \qquad \cup \left\{\|z- \theta\|_{\Phi(\theta)}^2 - \|\theta- z\|_{\Phi(z)}^2 \geq - \frac{2}{50} \right\}^c \cup \{ z \in K \}^c \bigg)\\
&\geq 1 - \mathbb{P}_{z \sim N(\theta, \Phi^{-1}(\theta))} \left( \left\{ \frac{\pi(z)}{\pi(\theta)} \geq \frac{1}{2} \right\}^c\right) - \mathbb{P}_{z \sim N(\theta, \Phi^{-1}(\theta))}\left(\left\{ \frac{\sqrt{\mathrm{det}(\Phi(z))}}{ \sqrt{\mathrm{det}(\Phi(\theta))}} \geq \frac{48}{50} \right \}^c\right)\\
&\qquad \qquad- \mathbb{P}_{z \sim N(\theta, \Phi^{-1}(\theta))}\left(\|z- \theta\|_{\Phi(\theta)}^2 - \|\theta- z\|_{\Phi(z)}^2 < - \frac{2}{50}\right) - \mathbb{P}_{z \sim N(\theta, \Phi^{-1}(\theta))}\left( \{ z \in K \}^c \right)\\
&\stackrel{\textrm{Lemmas }\ref{lemma_density_ratio}, \ref{lemma_det}, \ref{lemma_remain_in_ellipsoid}, \textrm{ Eq. } \ref{eq_f13b}}{\geq}  1- \frac{6}{10}- \frac{2}{100} - \frac{2}{100} - \frac{1}{100}\\
&\geq \frac{1}{3}.
\end{align*}
\end{proof}

\begin{lemma} \label{lemma_density_ratio}
Let $\theta \in \mathrm{int}(K)$. Then
\begin{equation*}
\mathbb{P}_{z \sim N(\theta, \Phi^{-1}(\theta))}\left( \frac{\pi(z)}
{\pi(\theta)} \geq \frac{99}{100} \right) \geq \frac{99}{100}.
\end{equation*}
\end{lemma}

\begin{proof}
Since  $z \sim N(\theta, \Phi(\theta)^{-1})$, we have that %
\begin{equation*}
    z = \theta + \Phi(\theta)^{-\frac{1}{2}}\xi = \theta + (\alpha^{-1} H(\theta) +  \eta^{-1}I_d)^{-\frac{1}{2}}
\end{equation*}
for some  $\xi \sim N(0,I_d)$.

Since $\alpha^{-1} H(\theta) +  \eta^{-1}I_d \succeq \eta^{-1}I_d$, and $H(\theta)$ and $I_d$ are both positive definite,  we have that
\begin{equation}\label{eq_g1}
    \eta I_d \succeq (\alpha^{-1} H(\theta) +  \eta^{-1}I_d)^{-1}. 
\end{equation}
Thus, 
\begin{align} \label{eq_g2}
   \|z-\theta\|_2 &= \|(\alpha^{-1} H(\theta) +  \eta^{-1}I_d)^{-\frac{1}{2}} \xi\|_2 \nonumber\\
   &= \sqrt{\xi^\top \left(   \alpha^{-1} H(\theta) +  \eta^{-1}I_d    \right)^{-1} \xi} \nonumber\\
   &\stackrel{\textrm{Eq. }\eqref{eq_g1}}{\leq}  \sqrt{\xi^\top \eta I_d \xi} \nonumber\\
   &= \sqrt{\eta} \|\xi\|_2.
\end{align}
Now, since $\xi \sim N(0,I_d)$, by the Hanson-wright concentration inequality for the $\chi$-distribution \cite{hanson1971bound}, we have that
\begin{equation}\label{eq_g3}
    \mathbb{P}(\|\xi\|_2 > t) \leq e^{-\frac{t^2-d}{8}} \qquad \forall t > \sqrt{2d}.
\end{equation}
Thus, Equations \eqref{eq_g2} and \eqref{eq_g3} together imply that
\begin{equation} \label{eq_g4}
    \mathbb{P}(\|z-\theta\|_{2} > \sqrt{\eta}\sqrt{40d}) \leq e^{-\frac{29d}{8}} < \frac{1}{100}.
\end{equation}

\noindent
Now, in the setting where $f$ is $L$-Lipschitz, we have
\begin{equation*}
    \frac{\pi(z)}{\pi(\theta)} = e^{-(f(z)-f(\theta))} \leq e^{-L\|z-\theta\|_2}.
\end{equation*}
Therefore,

\begin{eqnarray*}
\mathbb{P}_{z \sim N(\theta, \Phi^{-1}(\theta))}\left( \frac{\pi(z)}{\pi(\theta)} \geq \frac{99}{100} \right)
&\geq &\mathbb{P}_{z \sim N(\theta, \Phi^{-1}(\theta))}\left( e^{-L\|z-\theta\|_2} \geq \frac{99}{100} \right) \\
& = &\mathbb{P}_{z \sim N(\theta, \Phi^{-1}(\theta))}\left( \|z-\theta\|_2 \leq \frac{\log(\frac{100}{99})}{L} \right) \\
&\geq  &\frac{99}{100}
\end{eqnarray*}
where the last inequality holds by Inequality \eqref{eq_g4}, since $\eta \leq \frac{1}{10^4 d  L^2}$.

Moreover, in the setting where $f$ is differentiable and  $\beta$-Lipschitz, we have that, since $z-\theta$ is a multivariate Gaussian random variable, 
\begin{equation*}
\mathbb{P}((z-\theta)^\top \nabla f(\theta) \leq 0) = \frac{1}{2}. 
\end{equation*}
If $(z-\theta)^\top \nabla f(\theta) \leq 0$, we have that
\begin{align*}
    f(z) - f(\theta) &\leq (z-\theta)^\top \nabla f(\theta) + \beta \|z- \theta\|_2^2\\
    &\leq \beta \|z- \theta\|_2^2.
\end{align*}
Therefore,

\begin{align*}
&\mathbb{P}_{z \sim N(\theta, \Phi^{-1}(\theta))}\left( \frac{\pi(z)}{\pi(\theta)} \geq \frac{99}{100} \right)\\
&\geq \mathbb{P}_{z \sim N(\theta, \Phi^{-1}(\theta))}\left(\left\{ \frac{\pi(z)}{\pi(\theta)} \geq \frac{99}{100} \right\} \cap \left\{ (z-\theta)^\top \nabla f(\theta) \leq 0 \right\} \right) -  \mathbb{P}_{z \sim N(\theta, \Phi^{-1}(\theta))}\left((z-\theta)^\top \nabla f(\theta) > 0 \right)\\
&\geq \mathbb{P}_{z \sim N(\theta, \Phi^{-1}(\theta))}\left( e^{-\beta \|z- \theta\|_2^2} \geq \frac{99}{100} \right) - \frac{1}{2}\\
&= \mathbb{P}_{z \sim N(\theta, \Phi^{-1}(\theta))}\left( \|z-\theta\|_2 \leq \frac{\sqrt{\log(2)}}{\sqrt{\beta}} \right)\\
&\geq \frac{99}{10} - \frac{1}{2},
\end{align*}
where the last Inequality holds by Inequality \eqref{eq_g4}, since $\eta \leq \frac{1}{10^4 d  \beta}$.

\end{proof}

\begin{lemma} \label{Lemma_TV}

For any $\theta,z \in \mathrm{Int}(K)$ such that   $\|\theta-z\|_{\Phi(\theta)} \leq \frac{1}{4\alpha^{\nicefrac{1}{2}} }$, we have that

\begin{equation*}
\| N(\theta, \Phi^{-1}(\theta)) - N(z, \Phi^{-1}(z)) \|_{\mathrm{TV}}^2
\leq
3d \alpha \|\theta-z\|_{\Phi(\theta)}^2 +  \frac{1}{2}\|\theta-z\|_{\Phi(\theta)}^2
\end{equation*}

\end{lemma}

\noindent
The proof of Lemma \ref{Lemma_TV} is an adaptation of the proof Lemma 3 in \cite{sachdeva2016mixing} to our Markov chain's ``soft'' barrier function and follows roughly along the lines of that proof.
\begin{proof}
Since  $\|\theta-z\|_{\Phi(\theta)} \leq \frac{1}{4\alpha^{\nicefrac{1}{2}} }$, by Lemma \ref{lemma_PD} we have that
\begin{equation}  \label{eq_e2}
(1- \alpha^{\nicefrac{1}{2}}  \|\theta-z\|_{\Phi(\theta)})^2 \Phi(z) \preceq \Phi(\theta)  \preceq (1+ \alpha^{\nicefrac{1}{2}}  \|\theta-z\|_{\Phi(\theta)})^2 \Phi(z).
 \end{equation}    
 Therefore, since the product of any two positive definite matrices is also a positive definite matrix, Inequality \eqref{eq_e2} implies that
 \begin{equation*}
(1- \alpha^{\nicefrac{1}{2}}  \|\theta-z\|_{\Phi(\theta)})^2 I_d \preceq \Phi(z)^{-1} \Phi(\theta)  \preceq (1 + \alpha^{\nicefrac{1}{2}}  \|\theta-z\|_{\Phi(\theta)})^2 I_d.
 \end{equation*}    
Thus, denoting by $\lambda_i(M)$ the $i$th-largest eigenvalue of any matrix $M\in \mathbb{R}^d \times \mathbb{R}^d$, we have by Inequality \eqref{eq_e2} that
\begin{equation}\label{eq_e3}
(1- \alpha^{\nicefrac{1}{2}}  \|\theta-z\|_{\Phi(\theta)})^2 \leq \lambda_i(\Phi(z)^{-1}\Phi(\theta))  \leq (1 + \alpha^{\nicefrac{1}{2}}  \|\theta-z\|_{\Phi(\theta)})^2 \qquad \forall i \in [d].
\end{equation}    
Now the KL-divergence between any to multivariate Gaussian distributions with any means $\mu_1,\mu_2 \in \mathbb{R}^d$ and any covariance matrices $\Sigma_1,\Sigma_2 \in \mathbb{R}^d$ satisfies (see e.g. Section 9 of \cite{duchi2007derivations} or Fact 5 in \cite{sachdeva2016mixing})
\begin{equation} \label{eq_e18}
    D_{\mathrm{KL}}\left(N(\mu_1, \Sigma), N(\mu_2, \Sigma) \right)
    = \frac{1}{2} \left(\mathrm{Tr}(\Sigma_1^{-1}\Sigma_2)  -d + \log\left(\frac{\mathrm{det}(\Sigma_1)}{\mathrm{det}(\Sigma_2}\right) + (\mu_1-\mu_2)^\top\Sigma_1^{-1} (\mu_1-\mu_2)  \right).
\end{equation}
Therefore we have that
\begin{align*}
    &\| N(\theta, \Phi^{-1}(\theta)) - N(z, \Phi^{-1}(z) \|_{\mathrm{TV}}^2 \stackrel{\textrm{Pinsker's Inequality}}{\leq}  2D_{\mathrm{KL}}\left(N(\theta, \Phi^{-1}(\theta)), N(z, \Phi^{-1}(z)) \right)\\
    &\stackrel{\textrm{Eq. }\eqref{eq_e18}}{=} \frac{1}{2} \left(\mathrm{Tr}(\Phi(\theta) \Phi^{-1}(z))  -d + \log\left(\frac{\mathrm{det}(\Phi^{-1}(\theta))}{\mathrm{det}(\Phi^{-1}(z))}\right) + (\theta-z)^\top\Phi(\theta) (\theta-z)  \right)\\
    &= \frac{1}{2} \sum_{i=1}^{d}\lambda_i(\Phi(\theta) \Phi^{-1}(z))  - \frac{1}{2}d +  \frac{1}{2}\log\left(\frac{1}{\mathrm{det}(\Phi(\theta)\Phi^{-1}(z))}\right) +  \frac{1}{2}\|\theta-z\|_{\Phi(\theta)}^2\\
    &= \frac{1}{2}\sum_{i=1}^{d}\left( \lambda_i(\Phi(\theta) \Phi^{-1}(z)) - 1 + \log\left(\frac{1}{\lambda_i(\Phi(\theta) \Phi^{-1}(z)}\right) \right) +  \frac{1}{2}\|\theta-z\|_{\Phi(\theta)}^2\\
    &\leq \frac{1}{2}\sum_{i=1}^{d}\left( \lambda_i(\Phi(\theta) \Phi^{-1}(z)) + \frac{1}{\lambda_i(\Phi(\theta) \Phi^{-1}(z))} -2 \right) +  \frac{1}{2}\|\theta-z\|_{\Phi(\theta)}^2\\
         &\stackrel{\textrm{Eq. }\eqref{eq_e3}}{\leq} \frac{1}{2}\sum_{i=1}^{d}\left( \max_{t \in \left[\left(1- \alpha^{\nicefrac{1}{2}}  \|\theta-z\|_{\Phi(\theta)}\right)^2, \, \, \left(1+ \alpha^{\nicefrac{1}{2}}  \|\theta-z\|_{\Phi(\theta)}\right)^2\right]} t + \frac{1}{t} -2 \right) + \frac{1}{2}\|\theta-z\|_{\Phi(\theta)}^2\\
                  &= \frac{1}{2}\sum_{i=1}^{d}\left( \max_{t \in \left[- \alpha^{\nicefrac{1}{2}}  \|\theta-z\|_{\Phi(\theta)},\, \, \, \alpha^{\nicefrac{1}{2}}  \|\theta-z\|_{\Phi(\theta)}\right]} (1+t)^2 + \frac{1}{(1+t)^2} -2 \right) +  \frac{1}{2}\|\theta-z\|_{\Phi(\theta)}^2\\
                                    &\leq \frac{1}{2}\sum_{i=1}^{d}\left( \max_{t \in \left[- \alpha^{\nicefrac{1}{2}}  \|\theta-z\|_{\Phi(\theta)},\, \, \, \alpha^{\nicefrac{1}{2}}  \|\theta-z\|_{\Phi(\theta)}\right]} 6t^2 \right) +  \frac{1}{2}\|\theta-z\|_{\Phi(\theta)}^2\\
                &\stackrel{\textrm{convexity }}{\leq}  \frac{1}{2}\left( \sum_{i=1}^{d} 6 \alpha \|\theta-z\|_{\Phi(\theta)}^2 \right) +  \frac{1}{2}\|\theta-z\|_{\Phi(\theta)}^2\\
                &= 3d \alpha \|\theta-z\|_{\Phi(\theta)}^2 +  \frac{1}{2}\|\theta-z\|_{\Phi(\theta)}^2,
\end{align*}
where the first inequality is Pinsker's inequality, the second inequality holds because $\log(\frac{1}{t}) \leq \frac{1}{t}-1$ for all t>0, the fourth inequality holds because $(1+t)^2 + \frac{1}{(1+t)^2} -2 \leq t^2$ for all $t\in[-\frac{1}{4},\frac{1}{4}]$, and the fifth inequality holds since $t^2$ is convex for $t\in \mathbb{R}$.

\end{proof}

\noindent
For every $j \in \mathbb{N}$, let $m_j  = m + \lfloor \alpha \eta^{-1} j^2 \rfloor d$.
Consider the matrices $A^j = [A^\top, I_d, \ldots, I_d]^\top$ where $A$ is concatenated with $\frac{m_j -1}{d}$ copies of the identity matrix $I_d$.
And consider the vectors $b^j = (b^\top, j \textbf{1}^\top, \ldots, j \textbf{1}^\top)$, where $b$ is concatenated with $\frac{m_j -1}{d}$ copies of the vector $j \textbf{1}$,  \, \, where $\textbf{1}= (1,\ldots,1)^\top \in \mathbb{R}^d$ is the all-ones vector.
Then $K = \{\theta \in \mathbb{R}^d: A^j \theta \leq b^j\}$, and the Hessian of the corresponding log-barrier functions is
\begin{equation*}
    H_j(w) :=\sum_{i=1}^{m^j} \frac{a^j_i (a^j_i)^\top}{((a^j_i)^\top w -b^j_i)^2}.
\end{equation*}

\begin{lemma} \label{lemma_limits}
For all $w \in \mathrm{int}(K)$ we have that
\begin{equation} \label{eq_L1}
    \lim_{j\rightarrow \infty} H_j(w) = \alpha \Phi(w),
\end{equation}
uniformly in $w$, and that
\begin{equation} \label{eq_L2}
    \lim_{j\rightarrow \infty} (H_j(w))^{-1} = \alpha^{-1} (\Phi(w))^{-1},
\end{equation}
    uniformly in $w$.
Moreover, for any $\theta\in \mathrm{int}(K)$, any $z \in \frac{1}{2} D_{\theta}$ and any sequence $\{z_j\}_{j=1}^\infty \subseteq \frac{1}{2} D_{\theta}$ such that $\lim_{j\rightarrow \infty} z_j = z$ uniformly in $z$,  we have that
\begin{equation} \label{eq_L3}
    \lim_{j\rightarrow \infty} H_j(z_j) = \alpha \Phi(z),
\end{equation}
uniformly in $z$, and that
\begin{equation} \label{eq_L4}
    \lim_{j\rightarrow \infty}     \frac{\mathrm{det}(H_j(z_j))}{\mathrm{det}(H_j(\theta))} = \frac{\mathrm{det}(\Phi(z))}{\mathrm{det}(\Phi(\theta))},
\end{equation}
uniformly in $z$.
\end{lemma}

\begin{proof}

\begin{align} \label{eq_e4b}
H_j(w) &=\sum_{i=1}^{m} \frac{a_i a_i^\top}{(a_i^\top w -b_i)^2} +   \lfloor \alpha \eta^{-1} j^2 \rfloor  \sum_{i=1}^{d} \frac{e_i e_i^\top}{(e_i^\top w -j)^2}\nonumber \\
&=H(w) +  \lfloor \alpha \eta^{-1} j^2 \rfloor  \sum_{i=1}^{d} \frac{e_i e_i^\top}{(e_i^\top w -j)^2}.
\end{align}
Now, since $w \in K \subseteq B(0,R)$, we have that
\begin{equation}\label{eq_e7}
    (R -j)^2 \leq (e_i^\top w -j)^2 \leq (-R -j)^2.
\end{equation}
Thus,
\begin{equation}\label{eq_e5}
\lim_{j \rightarrow \infty} \frac{\lfloor \alpha^{-1} \eta^{-1} j^2 \rfloor}{(e_i^\top w -j)^2} \leq \lim_{j \rightarrow \infty} \frac{\lfloor \alpha \eta^{-1} j^2 \rfloor}{(R -j)^2} =  \alpha \eta^{-1}, 
\end{equation}
and
\begin{equation}\label{eq_e6}
\lim_{j \rightarrow \infty} \frac{\lfloor \alpha \eta^{-1} j^2 \rfloor}{(e_i^\top w -j)^2} \geq \lim_{j \rightarrow \infty} \frac{\lfloor \alpha^{-1} \eta^{-1} j^2 \rfloor}{(-R -j)^2} =  \alpha^{-1} \eta^{-1}.
\end{equation}
Thus, \eqref{eq_e4b} \eqref{eq_e5}, and \eqref{eq_e6} together imply that

\begin{align} \label{eq_e4}
\lim_{j \rightarrow \infty}  H_j(w) &= H(w) +  \sum_{i=1}^{d} e_i e_i^\top \times \lim_{j \rightarrow \infty}  \frac{\lfloor \alpha \eta^{-1} j^2 \rfloor}{(e_i^\top w -j)^2}\nonumber \\
&= H(w) +  \sum_{i=1}^{d} e_i e_i^\top \times  \alpha \eta^{-1}\nonumber \\
&=  H(w) +  \alpha \eta^{-1} I_d\nonumber \\
& = \alpha \Phi(w),
\end{align}
where Inequalities \eqref{eq_e5} and \eqref{eq_e6} imply that the convergence to the limit is uniform in $w$.
This proves Equation \eqref{eq_L1}.

Moreover, since $\{z_k\}_{k=1}^\infty \subseteq \frac{1}{2} D_{\theta}$, and $D_{\theta} \subseteq K$, we also have that 
\begin{equation} \label{eq_e15}
    |a_i^\top z_j -b_i| \geq \frac{1}{2}|a_i^\top \theta -b_i|.
\end{equation}
Therefore,
\begin{align*}
\lim_{j \rightarrow \infty}  H_j(z_j)
&= H(z_j) +  \sum_{i=1}^{d} e_i e_i^\top \times \lim_{j \rightarrow \infty}  \frac{\lfloor \alpha \eta^{-1} j^2 \rfloor}{(e_i^\top z_j -j)^2}\\
&= \sum_{i=1}^{m} \frac{a_i a_i^\top}{(a_i^\top z_j -b_i)^2} +  \sum_{i=1}^{d} e_i e_i^\top \times \lim_{j \rightarrow \infty}  \frac{\lfloor \alpha \eta^{-1} j^2 \rfloor}{(e_i^\top z_j -j)^2}\\
&= \sum_{i=1}^{m} \frac{a_i a_i^\top}{(a_i^\top z_j -b_i)^2} +  \sum_{i=1}^{d} e_i e_i^\top \times  \alpha \eta^{-1}\\
&\stackrel{\textrm{Eq. }\eqref{eq_e15}}{=}  \sum_{i=1}^{m} \frac{a_i a_i^\top}{(a_i^\top z -b_i)^2} +  \sum_{i=1}^{d} e_i e_i^\top \times  \alpha \eta^{-1}\\
&=  H(z) +  \alpha \eta^{-1} I_d\\
& = \alpha \Phi(z),
\end{align*}
where the convergence of the limit in the fourth equality holds uniformly in $z$ by \eqref{eq_e15} and the fact that $\lim_{j\rightarrow \infty} z_j = z$.
Thus, we have that
\begin{equation} \label{eq_e16}
    \lim_{j \rightarrow \infty}  H_j(z_j) = \alpha \Phi(z)
\end{equation}
uniformly in $z$.
This proves Equation \eqref{eq_L3}.

Moreover, since the determinant is a polynomial in the entries of the matrix, Inequality \eqref{eq_e4} implies that
\begin{equation}\label{eq_e11}
    \lim_{j \rightarrow \infty}  \mathrm{det}(H_j(w)) = \mathrm{det}(\alpha \Phi(w)),
\end{equation}
uniformly in $w$.

By Inequality \eqref{eq_e4b} we also have that

\begin{equation*}
H(w) +  \sum_{i=1}^{d} e_i e_i^\top \times  \frac{\lfloor \alpha \eta^{-1} j^2 \rfloor}{(R -j)^2} \leq
H_j(w) \leq
H(w) +  \sum_{i=1}^{d} e_i e_i^\top \times  \frac{\lfloor \alpha \eta^{-1} j^2 \rfloor}{(-R -j)^2}
\end{equation*}
and, hence, that
\begin{equation}  \label{eq_e8}
v^\top \left[ H(w) +  \frac{\lfloor \alpha \eta^{-1} j^2 \rfloor}{(R -j)^2} I_d \right] v \leq
v^\top H_j(w) v \leq
v^\top \left[ H(w) +  \frac{\lfloor \alpha \eta^{-1} j^2 \rfloor}{(-R -j)^2} I_d \right ] v \qquad \forall v \in \mathbb{R}^d.
\end{equation}
Thus, Inequality \eqref{eq_e8} implies that

\begin{equation}  \label{eq_e9}
H(w) +  \frac{\lfloor \alpha \eta^{-1} j^2 \rfloor}{(R -j)^2} I_d \preceq
 H_j(w) \preceq
 H(w) +  \frac{\lfloor \alpha \eta^{-1} j^2 \rfloor}{(-R -j)^2} I_d \qquad \forall j \in \mathbb{N}.
\end{equation}
Thus, Inequality \eqref{eq_e9} implies that

\begin{equation}  \label{eq_e17}
\left( H(w) +  \frac{\lfloor \alpha \eta^{-1} j^2 \rfloor}{(-R -j)^2} I_d  \right)^{-1} \preceq
 (H_j(w))^{-1} \preceq
 \left(H(w) +  \frac{\lfloor \alpha \eta^{-1} j^2 \rfloor}{(R -j)^2} I_d \right)^{-1} \qquad \forall j \in \mathbb{N}.
\end{equation}
Thus, since $H(w)$ is positive definite,  Inequality \eqref{eq_e17} together with inequalities \eqref{eq_e5} and \eqref{eq_e6} imply that
\begin{eqnarray*} 
    \lim_{j \rightarrow \infty}  (H_j(w))^{-1} &= &(H(w) +  \alpha \eta^{-1} I_d)^{-1}\\
    &=&\alpha^{-1} \Phi(w)^{-1} \qquad \forall w \in \mathrm{Int}(K)
\end{eqnarray*}
uniformly in $w$.
This proves Equation \eqref{eq_L2}.

\noindent
Moreover, Inequality \eqref{eq_e9} implies that
\begin{align}  \label{eq_e10}
 \mathrm{det}(H_j(w)) &\leq
\mathrm{det}\left(H(w) +  \frac{\lfloor \alpha \eta^{-1} j^2 \rfloor}{(R -j)^2} I_d\right)\nonumber \\
&\leq \left(\lambda_{\mathrm{max}}\left(H(w) +  \frac{\lfloor \alpha \eta^{-1} j^2 \rfloor}{(R -j)^2} I_d \right)\right)^d\nonumber \\
&\leq \left(\lambda_{\mathrm{max}}(H(w)) +  \lambda_{\mathrm{max}}\left(\frac{\lfloor \alpha \eta^{-1} j^2 \rfloor}{(R -j)^2} I_d) \right)\right)^d\nonumber \\
&\leq \left(\lambda_{\mathrm{max}}(H(w)) + 3 \alpha \eta^{-1}\right)^d\nonumber \\
&= c_1 \qquad \forall j \geq 3R,
\end{align}
for some $c_1>0$ which does not depend on $j$.

Inequality \eqref{eq_e8} also implies that
\begin{align}  \label{eq_e12}
 \mathrm{det}(H_j(w)) &\geq
\mathrm{det}\left(H(w) +  \frac{\lfloor \alpha \eta^{-1} j^2 \rfloor}{(R -j)^2} I_d\right)\nonumber \\
&\geq \left(\lambda_{\mathrm{min}}\left(H(w) +  \frac{\lfloor \alpha \eta^{-1} j^2 \rfloor}{(R -j)^2} I_d \right)\right)^d\nonumber \\
&\geq \left(\max\left(\lambda_{\mathrm{min}}(H(w)), \, \,  \lambda_{\mathrm{min}}\left(\frac{\lfloor \alpha \eta^{-1} j^2 \rfloor}{(R -j)^2} I_d) \right)\right)\right)^d\nonumber \\
&\geq \left(\frac{1}{3}\alpha \eta^{-1}\right)^d \nonumber \\ &=c_2 \qquad \forall j \geq 3R,
\end{align}
for some $c_2>0$ which does not depend on either $j$ or $w$.

Thus, Inequalities \eqref{eq_e11}, \eqref{eq_e10}, and \eqref{eq_e12} together imply that for any $\theta, z \in \mathrm{int}(K)$ we have that
\begin{align*}
    \lim_{j\rightarrow \infty}     \min \left( \frac{\mathrm{det}(H_j(z))}{\mathrm{det}(H_j(\theta))}, \, \, \, 1 \right )
    &\stackrel{\textrm{Eq. }\eqref{eq_e10},\eqref{eq_e12}}{=}   \min \left(\frac{\lim_{j\rightarrow \infty} \min(\mathrm{det}(H_j(z_j)), \, \frac{1}{c_2})}{\lim_{j\rightarrow \infty} \mathrm{det}(H_j(\theta))}, \, \, \, 1 \right )\\
    &\stackrel{\textrm{Eq. }\eqref{eq_e11}, \eqref{eq_e16}}{=}  \min \left(\frac{\mathrm{det}(\alpha\Phi(z))}{\mathrm{det}(\alpha\Phi(\theta))}, \, \, \, 1 \right )\\
    &=  \min \left(\frac{\mathrm{det}(\Phi(z))}{\mathrm{det}(\Phi(\theta))}, \, \, \, 1 \right ),
\end{align*}
where the limit converges uniformly in $z$.
This proves Equation \eqref{eq_L4}.

\end{proof}

\begin{lemma} \label{lemma_remain_in_ellipsoid}
Let $\xi \sim N(0,I_d)$ and let $\theta \in \mathrm{int}(K)$.  Then with probability at least $\frac{99}{100}$ we have that
\begin{equation*}
\theta + \alpha^{\frac{1}{2}} H^{-\nicefrac{1}{2}}(\theta)\xi \in \frac{1}{2} D_\theta
\end{equation*}
and
\begin{equation*}
    \|\xi\|_2 \leq 10 \sqrt{d}.
\end{equation*}

\end{lemma}

\begin{proof}
Let $z = \theta + \alpha^{\frac{1}{2}} H^{-\frac{1}{2}}(\theta)\xi$.
 Then since $\xi \sim N(0,I_d)$, by the Hanson-Wright concentration inequality \cite{hanson1971bound}, we have that
\begin{equation*}
    \mathbb{P}(\|\xi\|_2 > t) \leq e^{-\frac{t^2-d}{8}} \qquad \forall t > \sqrt{2d}.
\end{equation*}
And hence, since $\alpha^{-\nicefrac{1}{2}}\|z-\theta\|_{H(\theta)} =
\|H^{\frac{1}{2}}(\theta)H^{-\frac{1}{2}}(\theta)\xi\|_2 = \|\xi\|_2$, that
\begin{equation*}
    \mathbb{P}(\|z-\theta\|_{H(\theta)} > \alpha^{\nicefrac{1}{2}} t)  
    \leq e^{-\frac{t^2-d}{8}} \qquad \forall t > \sqrt{2d}.
\end{equation*}
Hence, 
\begin{equation*}
     \mathbb{P}(\|z-\theta\|_{H(\theta)} > \alpha^{\nicefrac{1}{2}}  10\sqrt{d})  
     \leq \frac{1}{100}.
\end{equation*}
Thus, since $\alpha \leq \frac{1}{100 d}$ we have that
\begin{equation*}
     \mathbb{P}\left(z \in \frac{1}{2} D_\theta \right)
     =\mathbb{P}\left(\|z-\theta\|_{H(\theta)} \leq \frac{1}{2}\right)  
 \geq \frac{99}{100}.
\end{equation*}
\end{proof}

\begin{lemma} \label{lemma_det}
Consider any $\theta \in \mathrm{int}(K)$, and $\xi \sim N(0,I_d)$.  Let $z= \theta + (\Phi(\theta))^{-\frac{1}{2}} \xi$.
Then
\begin{equation}\label{eq_f12}
    \mathbb{P}\left (\frac{\mathrm{det}\left(\Phi (z)\right)}{\mathrm{det}( \Phi(\theta))} \geq \frac{48}{50} \right) \geq \frac{98}{100},
\end{equation}
and

\begin{equation}\label{eq_f13}
\mathbb{P}\left( \|z - \theta\|_{\Phi(z)}^2 - \|z - \theta\|_{\Phi(\theta)}^2
\leq \frac{2}{50} \right )\geq \frac{98}{100}.
\end{equation}

\end{lemma}

\begin{proof} 
Let $z_j = \theta+ \alpha^{\frac{1}{2}} H_j^{-\frac{1}{2}}(\theta) \xi$ for all $j \in \mathbb{N}$.
Since $H_j(\theta) \succeq H(\theta)$ for all $j \in \mathbb{N}$, we have that $z_j = \theta+ \alpha^{\frac{1}{2}} H_j^{-\frac{1}{2}}(\theta) \in \frac{1}{2} D_{\theta}$ whenever $\theta+ \alpha^{\frac{1}{2}} H(\theta)^{-\frac{1}{2}} \in \frac{1}{2} D_{\theta}$. 
Let $E$ be the event that $\|\xi\|_2 \leq 10 \sqrt{d}$ and that 
$\{z_j\}_{j=1}^{\infty} \subseteq \frac{1}{2}D_{\theta}$.
Thus, by Lemma \ref{lemma_remain_in_ellipsoid}, we have that 
\begin{equation}\label{eq_f1}
\mathbb{P}\left( E \right) \geq \frac{99}{100}.
\end{equation}
Moreover, by Equation \eqref{eq_L2} of Lemma \ref{lemma_limits} we have that  $\lim_{j\rightarrow \infty} H_j^{-1}(\theta) = \alpha^{-1} \Phi^{-1}(\theta)$,
which implies that
\begin{align}\label{eqf_f5}
    \lim_{j\rightarrow \infty} z_j 
&= \lim_{j\rightarrow \infty} \theta + \alpha^{\frac{1}{2}} H_j^{-\frac{1}{2}} (\theta) \xi\nonumber \\
  &= \theta + \Phi^{-\frac{1}{2}}(\theta) \xi\nonumber \\
  &= z
\end{align}
uniformly in $\xi$ (and hence uniformly in $z =  \theta + \Phi^{-\frac{1}{2}}(\theta) \xi$) whenever the event $E$ occurs.
Therefore, by Equation \eqref{eq_L4} of Lemma \ref{lemma_limits} we have that 
\begin{equation}\label{eqf_f3}
    \lim_{j\rightarrow \infty}     \frac{\mathrm{det}(H_j(z_j))}{\mathrm{det}(H_j(\theta))} = \frac{\mathrm{det}(\Phi(z))}{\mathrm{det}(\Phi(\theta))},
\end{equation}
uniformly in $z$ whenever the event $E$ occurs.

Since, for each $j\in \mathbb{N}$, $H_j$ is the Hessian of a log-barrier function for $K$, by Proposition 6 in \cite{sachdeva2016mixing}, for all $t \in (0, \frac{1}{2}]$, all $\gamma \leq \frac{t}{\sqrt{2 \log(\frac{1}{t})}}$, and all $j \in \mathbb{N}$ we have that
\begin{equation*}
    \mathbb{P}_{\xi \sim N(0,I_d)}\left(  \frac{\mathrm{det}(H_j(\theta+\frac{\gamma}{\sqrt{d}}H_j(\theta)^{-\frac{1}{2}}\xi))}{ \mathrm{det}(H_j(\theta))} \geq e^{-2 t}\right) \geq 1-t.
\end{equation*}
Setting $t=\frac{1}{10}$, we have
\begin{equation} \label{eq_f0}
    \mathbb{P}_{\xi \sim N(0,I_d)}\left(  \frac{\mathrm{det}(H_j(\theta+ \alpha^{\frac{1}{2}} H_j(\theta)^{-\frac{1}{2}} \xi))}{ \mathrm{det}(H_j(\theta))} \geq \frac{49}{50}\right) \geq \frac{99}{100},
\end{equation}
since $\alpha \leq \frac{1}{400 d}$.
Inequalities \eqref{eq_f1}  and \eqref{eq_f0} together imply that
\begin{equation} \label{eq_f2}
    \mathbb{P}_{\xi \sim N(0,I_d)}\left(\left\{  \frac{\mathrm{det}(H_j(z_j))}{ \mathrm{det}(H_j(\theta))} \geq \frac{49}{50} \right\} \cap E  \right) \geq \frac{98}{100} \qquad \forall j \in \mathbb{N}.
\end{equation}
Moreover, Equation \eqref{eqf_f3} implies that there exists some number $N\in \mathbb{N}$ such that
\begin{equation} \label{eq_f4}
\left\{  \frac{\mathrm{det}(H_N(z_N))}{ \mathrm{det}(H_N(\theta))} \geq \frac{49}{50} \right\} \cap E \subseteq  \left\{  \frac{\mathrm{det}(\Phi(z))}{ \mathrm{det}(\Phi(\theta))} \geq \frac{48}{50} \right\} \cap E.
\end{equation}
Hence,
\begin{align} \label{eq_f5}
    \frac{98}{100} &\stackrel{\textrm{Eq. }\eqref{eq_f2}}{\leq}  \mathbb{P}_{\xi \sim N(0,I_d)}\left(\left\{  \frac{\mathrm{det}(H_N(z_N))}{ \mathrm{det}(H_N(\theta))} \geq \frac{49}{50} \right\} \cap E  \right)\nonumber \\
    &\stackrel{\textrm{Eq. }\eqref{eq_f4}}{\leq}  \mathbb{P}_{\xi \sim N(0,I_d)}\left( \left\{  \frac{\mathrm{det}(\Phi(z))}{ \mathrm{det}(\Phi(\theta))} \geq \frac{48}{50} \right\} \cap E\right).
\end{align}
This proves Inequality \eqref{eq_f12}.

Moreover,  since, for each $j\in \mathbb{N}$, $H_j$ is the Hessian of a log-barrier function for $K$, by Proposition 7 in \cite{sachdeva2016mixing}, for all $t \in (0, \frac{1}{2}]$, all $\alpha > 0$ such that 
$\sqrt{\alpha d} \leq \frac{t}{20} \log ( \frac{11}{t} )^{-\frac{3}{2}}$, and all $j \in \mathbb{N}$, we have that
    \begin{equation}\label{eq_f8}
         \mathbb{P}\left(\|z_j - \theta\|_{H_j(z_j)}^2 - \|z_j - \theta\|_{H_j(\theta)}^2 \geq 2t \alpha \right) \leq 1- t  \qquad \forall j \in \mathbb{N}.
    \end{equation}
    Thus, Equations \eqref{eq_f8} and \eqref{eq_f1} imply that
        \begin{equation}\label{eq_f9}
         \mathbb{P}\left(\left\{\|z_j - \theta\|_{H_j(z_j)}^2 - \|z_j - \theta\|_{H_j(\theta)}^2 \leq \frac{1}{50} \alpha \right \} \cap E \right) \geq \frac{98}{100}  \qquad \forall j \in \mathbb{N},
    \end{equation}
    since $\alpha \leq \frac{1}{10^5 d}$.
    
    By Equation \eqref{eq_L3}  of Lemma \ref{lemma_limits}, Equation \eqref{eqf_f5} implies that
\begin{equation}\label{eq_f6}
    \lim_{j\rightarrow \infty} H_j(z_j) = \alpha \Phi(z),    
\end{equation}
uniformly in $z$, whenever the event $E$ occurs.
    Thus, Equation \eqref{eq_f6} implies that
    \begin{align}\label{eq_f7}
    \lim_{j \rightarrow \infty} \|z_j - \theta\|_{H_j(z_j)}^2 - \|z_j - \theta\|_{H_j(\theta)}^2
    &=   \lim_{j \rightarrow \infty}  (z_j -\theta)^\top H_j(z_j) (z_j -\theta) - (z_j -\theta)^\top H_j(\theta) (z_j -\theta)\nonumber \\
    &\stackrel{\textrm{Eq. }\eqref{eq_f6}}{=}    \lim_{j \rightarrow \infty}  (z -\theta)^\top \alpha \Phi(z) (z -\theta) - (z -\theta)^\top \alpha \Phi(\theta) (z -\theta)\nonumber \\
    &=\alpha \|z - \theta\|_{\Phi(z)}^2 - \alpha\|z - \theta\|_{\Phi(\theta)}^2
    \end{align}
    uniformly in $z$ (and hence in $\xi =  \Phi^{\frac{1}{2}}(\theta)(z - \theta)$) whenever event $E$ occurs.
    Thus, Equation \eqref{eq_f7} implies that there exists a number $M \in \mathbb{N}$ such that
    \begin{equation} \label{eq_f10}
        \left\{ \|z_M - \theta\|_{H_M(z_M)}^2 - \|z_M - \theta\|_{H_M(\theta)}^2 \leq \frac{1}{50} \alpha \right \} \cap E \subseteq  \left\{ \alpha \|z - \theta\|_{\Phi(z)}^2 - \alpha\|z - \theta\|_{\Phi(\theta)}^2 \leq \frac{2}{50} \alpha \right \} \cap E.
    \end{equation}
   Thus,
   \begin{align*}
        \frac{98}{100} &\stackrel{\textrm{Eq. }\eqref{eq_f9}}{\leq} \mathbb{P}\left( \left\{ \|z_M - \theta\|_{H_M(z_M)}^2 - \|z_M - \theta\|_{H_M(\theta)}^2 \leq \frac{1}{50} \alpha \right \} \cap E \right )\\
        &\stackrel{\textrm{Eq. }\eqref{eq_f10}}{\leq}  \mathbb{P}\left( \left\{ \alpha \|z - \theta\|_{\Phi(z)}^2 - \alpha\|z - \theta\|_{\Phi(\theta)}^2 \leq \frac{2}{50} \alpha \right \} \cap E \right ).
   \end{align*}
    This proves Inequality \eqref{eq_f13}.
    
\end{proof}

\subsection{Bounding the conductance}

\begin{lemma} \label{Lemma_transition_TV}
For any $\theta, z \in \mathrm{int}(K)$  we have that
\begin{equation*}
    \|P_\theta-P_z\|_{\mathrm{TV}}\leq\frac{29}{30},  \qquad \mathrm{whenever} \qquad \|\theta-z\|_{\Phi(\theta)}\leq 1.
\end{equation*}

\end{lemma}

\begin{proof}
First, we note that,
\begin{equation}\label{eq_z1}
\|P_\theta-P_z\|_{\mathrm{TV}}\leq\|P_\theta-N(\theta,\Phi^{-1}(\theta))\|_{\mathrm{TV}}+\|N(\theta,\Phi^{-1}(\theta))-N(z,\Phi^{-1}(z))\|_{\mathrm{TV}}+\|P_z-N(z,\Phi^{-1}(z)) \|_{\mathrm{TV}}.
\end{equation}
By Lemma \ref{Lemma_TV}, the middle term on the r.h.s. of \eqref{eq_z1} satisfies
\begin{equation}\label{eq_z4}
    \| N(\theta,\Phi^{-1}(\theta))-N(z,\Phi^{-1}(z))\|_{\mathrm{TV}} \leq\sqrt{3d\alpha\|\theta-z\|_{\Phi(\theta)}^2+\frac{1}{2}\|\theta-z\|_{\Phi(\theta)}^2}
\end{equation}
Plugging in our choice of hyperparameter $\alpha=\frac{1}{10^5 d}$, Inequality \eqref{eq_z4} simplifies to
\begin{equation}\label{eq_z2}
\|N(\theta,\Phi^{-1}(\theta))-N(z,\Phi^{-1}(z))\|_{\mathrm{TV}}\leq\sqrt{\frac{53}{100}}\|\theta-z\|_{\Phi(\theta)}.
\end{equation}
Thus, if we can show that $\|P_\theta-N(\theta,\Phi^{-1}(\theta))\|_{\mathrm{TV}}\leq\frac{1}{5}$ for all $\theta\in\mathrm{int}(K)$, we will have that $\|P_\theta-P_z\|_{\mathrm{TV}}\leq\frac{29}{30}$ whenever $\|\theta-z\|_{\Phi(\theta)}\leq\frac{1}{2}$, as desired.

To bound the other two terms on the r.h.s. of \eqref{eq_z1}, we observe that
\begin{equation}\label{eq_z5}
\|P_\theta-N(\theta,\Phi^{-1}(\theta))\|_{\mathrm{TV}}=1-\mathbb{E}_{z\sim N(\theta,\Phi^{-1}(\theta))}[q(\theta,z)],
\end{equation}
where $$q(\theta,z):=\min\{1,\frac{\pi(z)}{\pi(\theta)}\frac{\sqrt{\mathrm{det}\Phi(z)}}{\sqrt{\mathrm{det}\Phi(\theta)}}\exp(\|z-\theta\|_{\Phi(\theta)}^2-\|\theta-z\|_{\Phi(z)}^2)\times\mathbbm{1}\{z\in K\}\}$$ is the acceptance ratio.

By Lemma \ref{lemma_density_ratio}, we have that
$\mathbb{P}_{z\sim N(\theta,\Phi^{-1}(\theta))}\left(\frac{\pi(z)}{\pi(\theta)}\geq\frac{99}{100}\right)\geq\frac{99}{10}.$
Therefore,

\begin{align*}
&\mathbb{P}_{z\sim N(\theta,\Phi^{-1}(\theta))}[q(\theta,z)\geq\frac{9}{10}]\\
&\stackrel{\textrm{Lemma }\ref{lemma_density_ratio}}{\geq}  1-\mathbb{P}_{z\sim N(\theta,\Phi^{-1}(\theta))}\left(\frac{\pi(z)}{\pi(\theta)}<\frac{99}{100}\right)-\mathbb{P}_{z\sim N(\theta,\Phi^{-1}(\theta))}\left(\frac{\sqrt{\mathrm{det}(\Phi(z))}}{\sqrt{\mathrm{det}(\Phi(\theta))}}<\sqrt{\frac{48}{50}}\right)\\
&-\mathbb{P}_{z\sim N(\theta,\Phi^{-1}(\theta))}\left(e^{\|z-\theta\|_{\Phi(\theta)}^2-\|\theta-z\|_{\Phi(z)}^2}<0.96 \right)-\mathbb{P}_{z\sim N(\theta,\Phi^{-1}(\theta))}(z\notin K)\\
&\geq 1-\frac{1}{100}-\frac{2}{100}-\frac{2}{100}-\frac{1}{100}\\
&\geq\frac{9}{10},
\end{align*}
where the term with the exponent is bounded by Inequality \eqref{eq_f13} of Lemma \ref{lemma_det}, and the other two terms are bounded by Lemmas \ref{lemma_remain_in_ellipsoid} and \ref{lemma_det}.
Therefore,
\begin{equation}\label{eq_z3}
\|P_\theta-N(\theta,\Phi^{-1}(\theta)) \|_{\mathrm{TV}}=1-\mathbb{E}_{z\sim N(\theta,\Phi^{-1}(\theta))}[q(\theta,z)\geq\frac{9}{10}]\geq 1-\frac{9}{10}\times\frac{9}{10}\geq \frac{4}{5}.
\end{equation}
Plugging Inequalities \eqref{eq_z2} and \eqref{eq_z3} into \eqref{eq_z1}, we obtain that  $\| P_\theta-P_z\|_{\mathrm{TV}}\leq\frac{29}{30}$ whenever $\|\theta-z\|_{\Phi(\theta)}\leq\frac{1}{2}$, as desired.
\end{proof}

\noindent
We recall the following isoperimetric inequality for a log-concave distribution on a convex body, which uses the cross-ratio distance:
\begin{lemma}[Isoperimetric inequality for cross-ratio distance (Theorem 2.2 of \cite{lovasz2003hit})]\label{lemma_isoperimetric}
Let $\pi: \mathbb{R}^d \rightarrow \mathbb{R}$ be a log-concave density, with support on a convex body $K$. Then for any partition of $\mathbb{R}^d$ 
 into measurable sets $S_1,S_2,S_3$, the induced measure $\pi^\star$ satisfies
\begin{equation*}
    \pi^\star(S_3) \geq \sigma(S_1, S_2) \pi^\star(S_1) \pi^\star(S_2).
\end{equation*}
\end{lemma}

\noindent
For any $\theta \in \mathrm{Int}(K)$, define the random variable $Z_\theta$ to be the step taken by the Markov chain in Algorithm \ref{alg_Soft_Dikin_Walk} from the point $\theta$.
That is, 
set $z = \theta + \Phi(\theta)^{-\frac{1}{2}} \xi$  where $\xi \sim N(0,I_d)$.
 If  $z\in K$, set $Z_\theta = z$ with probability $\min \left (\frac{\pi(z) \sqrt{\mathrm{det}(\Phi(z))}}
{\pi(\theta) \sqrt{\mathrm{det}(\Phi(\theta))}} \exp\left(\|z- \theta\|_{\Phi(\theta)}^2 - \|\theta- z\|_{\Phi(z)}^2\right)
 , 1 \right)$.
 Else, set $z = \theta$.

For any $\theta$, $S \subseteq \mathbb{R}^d$, define the one-step transition probability of the soft-threshold Dikin walk Markov chain to be
\begin{equation*}
P_\theta(S) = \mathbb{P}(Z_\theta \in S).
\end{equation*}
The next proposition shows that the soft-threshold Dikin walk Markov chain is reversible and $\pi$ is a stationary distribution of this Markov chain:
\begin{proposition}[Reversibility and stationary distribution]\label{prop_stationary}
For any $S_1,S_2 \subseteq K$ we have that
\begin{equation*}
   \int_{\theta \in S_1} \pi(\theta) P_\theta(S_2) \mathrm{d} \theta =    \int_{z \in S_2} \pi(z) P_z(S_1) \mathrm{d} z.
\end{equation*}
\end{proposition}
\begin{proof}Let $\rho_\theta(z) := \frac{\sqrt{\mathrm{det}(\Phi(\theta))}}{(2 \mathrm{pi})^{\frac{d}{2}}} e^{-\frac{1}{2}(\theta - z)^\top \Phi(\theta) (\theta - z)}$ for any $\theta, z \in \mathrm{Int}(K)$ be the density of the $N(\theta,\Phi(\theta)^{-1})$ distribution.

\begin{align*}
     \int_{\theta \in S_1} \pi(\theta) P_\theta(S_2) \mathrm{d} \theta
     &=    \int_{\theta \in S_1} \pi(\theta) \int_{z \in S_2} \rho_\theta(z) \min \left (\frac{\pi(z) \sqrt{\mathrm{det}(\Phi(z))}}
{\pi(\theta) \sqrt{\mathrm{det}(\Phi(\theta))}} \exp\left(\|z- \theta\|_{\Phi(\theta)}^2 - \|\theta- z\|_{\Phi(z)}^2\right)
 , 1 \right)  \mathrm{d} z \mathrm{d} \theta\\
&=     \int_{\theta \in S_1} \int_{z \in S_2} \pi(\theta) \rho_\theta(z) \min \left (\frac{\pi(z)  \rho_z(\theta)}
{\pi(\theta)  \rho_\theta(z)} , 1 \right)  \mathrm{d} z \mathrm{d} \theta\\
&=    \int_{z \in S_2}\int_{\theta \in S_1}  \pi(\theta) \rho_z(\theta) \min \left (\frac{\pi(\theta)  \rho_\theta(z)}
{\pi(z)  \rho_z(\theta)} , 1 \right)  \mathrm{d} \theta \mathrm{d} z\\
&=  \int_{\theta \in S_1} \pi(\theta) P_\theta(S_2) \mathrm{d} \theta.
\end{align*}
\end{proof}

\noindent
Define the conductance $\phi$ of the Markov chain to be
\begin{equation*}
    \phi = \inf_{S \subseteq K : \pi^\star(S)\leq \frac{1}{2}} \frac{1}{\pi^\star(S)} \int_S P_\theta (K \backslash S) \pi(\theta) \mathrm{d} \theta.
\end{equation*}

\begin{lemma}\label{lemma_conductance}
The conductance $\phi$ satisfies
\begin{equation*}
    \phi \geq  \frac{1}{10^4} \frac{1}{\sqrt{2m \alpha^{-1} + \eta^{-1} R^{2}}}.
\end{equation*}
\end{lemma}

\begin{proof} The proof follows the general format for conductance proofs for geometric Markov chains.
Let $S_1 \subseteq K$ and let $S_2 = K \backslash S_1$.
Without loss of generality, we may assume that $\pi(S_1) \leq \frac{1}{2}$ (since otherwise we could just swap the names ``$S_1$'' and ``$S_2$'').
Let
\begin{align} \label{eq_h2}
    S_1' &= \left\{\theta \in S_1 : P_\theta(S_2) \leq \frac{1}{70}\right\}, \nonumber\\
    S_2' &= \left\{z \in S_2 : P_z(S_1) \leq \frac{1}{70} \right\},
\end{align}
and let
\begin{equation*}
    S_3' = (K \backslash S_1')\backslash S_2'.
\end{equation*}
By Proposition \ref{prop_stationary} we have that
\begin{equation}\label{eq_h6}
   \int_{\theta \in S_1} \pi(\theta) P_\theta(S_2) \mathrm{d} \theta =    \int_{\theta \in S_2} \pi(\theta) P_\theta(S_1) \mathrm{d} \theta.
\end{equation}
Moreover, by Lemma \ref{Lemma_transition_TV}, for any $u,v \in \mathrm{Int}(K)$ we have that
%
\begin{equation}\label{eq_h1}
\| P_{u} - P_v \|_{\mathrm{TV}} \leq \frac{29}{30} \textrm{ whenever } \|u-v\|_{\Phi(u)} \leq \frac{1}{2}.
\end{equation}
Thus, on the one hand, by Lemma \ref{lemma_cross_ratio}, Inequality \eqref{eq_h1} implies that for any $u,v \in \mathrm{Int}(K)$,
\begin{equation} \label{eq_h3}
\| P_{u} - P_v \|_{\mathrm{TV}} \leq \frac{29}{30}, \qquad \textrm{ whenever }  \qquad  \sigma(u, v) \leq \frac{1}{2\sqrt{2m \alpha^{-1} + \eta^{-1} R^{2}}}.
\end{equation}
On the other hand, Equations \eqref{eq_h2} imply that, for any $u \in S_1'$, $v \in S_2'$ we have that
\begin{equation}\label{eq_h4}
    \|P_u - P_v \|_{\mathrm{TV}} \geq  1 - P_u(S_2) - P_v(S_1) \geq \frac{68}{70}.
\end{equation}
Thus, Inequalities \eqref{eq_h3} and \eqref{eq_h4} together imply that
\begin{equation} \label{eq_h5}
\sigma(S_1', S_2') > \frac{1}{2\sqrt{2m \alpha^{-1} + \eta^{-1} R^{2}}}.
\end{equation}
Moreover by Lemma \ref{lemma_isoperimetric} we have that
\begin{equation}\label{eq_h7}
    \pi^\star(S_3') \geq \sigma(S_1', S_2') \pi^\star(S_1') \pi^\star(S_2').
\end{equation}
First, we assume that both $\pi^\star(S_1') \geq \frac{1}{4} \pi^\star(S_1)$ and $\pi^\star(S_2') \geq \frac{1}{4} \pi^\star(S_2)$.
In this case we have
\begin{align}\label{eq_h8}
    \int_{S_1} P_\theta(S_2) \pi(\theta) \mathrm{d}\theta &\stackrel{\textrm{Eq. }\ref{eq_h6}}{=}    \frac{1}{2}\int_{S_1} P_\theta(S_2) \pi(\theta) \mathrm{d}\theta + \frac{1}{2}\int_{S_2} P_\theta(S_1) \pi(\theta) \mathrm{d}\theta \nonumber \\
    &\stackrel{\textrm{Eq. }\ref{eq_h2}}{\geq}   \frac{1}{140} \pi^\star(S_3')\nonumber \\
    & \stackrel{\textrm{Eq. }\ref{eq_h7}}{\geq}    \frac{1}{140} \sigma(S_1', S_2') \pi^\star(S_1') \pi^\star(S_2)\nonumber \\
    & \stackrel{\textrm{Eq. }\ref{eq_h5}}{\geq}    \frac{1}{280} \frac{1}{\sqrt{2m \alpha^{-1} + \eta^{-1} R^{2}}} \pi^\star(S_1') \pi^\star(S_2')\nonumber \\
    &  \geq \frac{1}{560} \frac{1}{\sqrt{2m \alpha^{-1} + \eta^{-1} R^{2}}} \min(\pi^\star(S_1'), \,\,\, \pi^\star(S_2'))\nonumber \\
    & \geq \frac{1}{2024} \frac{1}{\sqrt{2m \alpha^{-1} + \eta^{-1} R^{2}}} \min(\pi^\star(S_1), \,\,\, \pi^\star(S_2)).
\end{align}
Now suppose that instead either $\pi^\star(S_1') < \frac{1}{4} \pi^\star(S_1)$ or $\pi^\star(S_2') < \frac{1}{4} \pi^\star(S_2)$.
If  $\pi^\star(S_1') < \frac{1}{4} \pi^\star(S_1)$  then we have
\begin{align}\label{eq_h9}
\int_{S_1} P_\theta(S_2) \pi(\theta) \mathrm{d}\theta &\stackrel{\textrm{Eq. }\ref{eq_h6}}{=}    \frac{1}{2}\int_{S_1} P_\theta(S_2) \pi(\theta) \mathrm{d}\theta + \frac{1}{2}\int_{S_2} P_\theta(S_1) \pi(\theta) \mathrm{d}\theta\nonumber \\
&\geq  \frac{1}{2}\int_{S_1\backslash S_1'} P_\theta(S_2) \pi(\theta) \mathrm{d}\theta\nonumber \\
  &\stackrel{\textrm{Eq. }\ref{eq_h2}}{\geq}   \frac{1}{2}\times \frac{3}{4} \times \frac{34}{70}\pi^\star(S_1)\nonumber \\
  & \geq \frac{1}{10} \min(\pi^\star(S_1), \,\,\, \pi^\star(S_2)).
\end{align}
Similarly, if  $\pi^\star(S_2') < \frac{1}{4} \pi^\star(S_2)$  we have
\begin{align}\label{eq_h10}
\int_{S_1} P_\theta(S_2) \pi(\theta) \mathrm{d}\theta &\stackrel{\textrm{Eq. }\ref{eq_h6}}{=}    \frac{1}{2}\int_{S_1} P_\theta(S_2) \pi(\theta) \mathrm{d}\theta + \frac{1}{2}\int_{S_2} P_\theta(S_1) \pi(\theta) \mathrm{d}\theta\nonumber \\
&\geq  \frac{1}{2}\int_{S_2\backslash S_2'} P_\theta(S_1) \pi(\theta) \mathrm{d}\theta\nonumber \\
  &\stackrel{\textrm{Eq. }\ref{eq_h2}}{\geq}   \frac{1}{2}\times \frac{3}{4} \times \frac{34}{70}\pi^\star(S_2)\nonumber \\
  & \geq \frac{1}{10} \min(\pi^\star(S_1), \,\,\, \pi^\star(S_2)).
\end{align}
Therefore, Inequalities \eqref{eq_h8}, \eqref{eq_h9}, and \eqref{eq_h10} together imply that
\begin{equation} \label{eq_h11}
\frac{1}{\min(\pi^\star(S_1), \,\,\, \pi^\star(S_2))}   \int_{S_1} P_\theta(S_2) \pi(\theta) \mathrm{d}\theta \geq \frac{1}{10^4} \frac{1}{\sqrt{2m \alpha^{-1} + \eta^{-1} R^{2}}}.
\end{equation}
for every partition $S_1\cup S_2 = K$.
Hence, Inequality \eqref{eq_h11} implies that
\begin{equation*}
    \phi = \inf_{S \subseteq K : \pi^\star(S)\leq \frac{1}{2}} \frac{1}{\pi^\star(S)} \int_S P_\theta (K \backslash S) \pi(\theta) \mathrm{d} \theta  \geq \frac{1}{60} \frac{1}{\sqrt{2m \alpha^{-1} + \eta^{-1} R^{2}}}.
\end{equation*}
\end{proof}

\noindent

\begin{definition}
We say that a distribution $\nu$ is $w$-warm for some $w \geq 1$ with respect to the stationary distribution $\pi$ if $\sup_{z\in K} \frac{\nu(z)}{\pi(z)} \leq w$.
\end{definition}

\begin{lemma}[Corollary 1.5 of \cite{lovasz1993random}]\label{lemma_cheeger}
Suppose that $\mu_0$ is the initial distribution of a lazy reversible Markov chain with conductance $\phi>0$ and stationary distribution $\pi$, and let $\mu_i$ be the distribution of this Markov chain after $i\geq 0$ steps. Suppose that  $\mu_0$ is $w$-warm with respect to $\pi$ for some $w \geq 1$.
Then for all $i\geq 0$ we have
\begin{equation*}
 \| \mu_i - \pi \|_{\mathrm{TV}} \leq \sqrt{w}\left(1- \frac{\phi^2}{2}\right)^i.  
\end{equation*}

\end{lemma}

\begin{lemma}\label{Lemma_TV_output}
Let $\delta >0$.  Suppose that $f: K \rightarrow \mathbb{R}$ is either $L$-Lipschitz (or has $\beta$-Lipschitz gradient). Suppose that $\theta_0 \sim \mu_0$ where $\mu_0$ is a $w$-warm  distribution with respect to $\pi \propto e^{-f}$ with support on $K$. 
Let $\mu$ denote the distribution of the output of Algorithm \ref{alg_Soft_Dikin_Walk} with hyperparameters $\alpha \leq \frac{1}{10^5 d}$ and $\eta \leq \frac{1}{10^4 d  L^2}$ (or $\eta \leq \frac{1}{10^4 d  \beta}$), if it is run for $T$ iterations. 
Then if $T\geq 10^9 \left( 2m \alpha^{-1} + \eta^{-1} R^{2} \right) \times \log(\frac{w}{\delta})$ we have that
\begin{equation*}
 \| \mu - \pi \|_{\mathrm{TV}} \leq \delta.  
\end{equation*}

\end{lemma}

\begin{proof} By Lemma \ref{lemma_conductance} we have that the conductance $\phi$ of the Markov chain in Algorithm \ref{alg_Soft_Dikin_Walk} satisfies
\begin{equation*}
    \phi \geq  \frac{1}{10^4} \frac{1}{\sqrt{2m \alpha^{-1} + \eta^{-1} R^{2}}},
\end{equation*}
and hence that
\begin{equation*}
T= 10^9 \left( 2m \alpha^{-1} + \eta^{-1} R^{2} \right) \times \log(\frac{w}{\delta}) \geq 2\phi^{-2}  \times \log(\frac{w}{\delta}).
\end{equation*}

\noindent
Thus, by Lemma \ref{lemma_conductance} we have that
\begin{align*}
 \| \mu_T - \pi \|_{\mathrm{TV}} &\leq \sqrt{w}\left(1- \frac{\phi^2}{2}\right)^T\\
 & = \sqrt{w}\left(1- \frac{\phi^2}{2}\right)^{2\phi^{-2} \log(\frac{w}{\delta})}\\
 &\leq \sqrt{w} e^{-\log(\frac{w}{\delta})}\\
  &\leq  \delta.
\end{align*}
\end{proof}

\begin{proof}\textbf{[of Theorem \ref{thm_soft_threshold_Dikin}]}

\paragraph{Total variation bound.}
Recall that we have set the step size parameters $\alpha = \frac{1}{10^5 d}$ and either $\eta = \frac{1}{10^4 d  L^2}$ (if $f$ is $L$-Lipschitz) or $\eta = \frac{1}{10^4 d  \beta}$   (if $f$ is $\beta$-smooth).
Thus, after running Algorithm \ref{alg_Soft_Dikin_Walk} for $T = 10^9 \left( 2m \alpha^{-1} + \eta^{-1} R^{2} \right) \times \log(\frac{w}{\delta})$ iterations, by Lemma \ref{Lemma_TV_output} we have that the distribution $\mu$ of the output of Algorithm \ref{alg_Soft_Dikin_Walk} satisfies 
\begin{equation}
 \| \mu - \pi \|_{\mathrm{TV}} \leq \delta.  
\end{equation}

\paragraph{Bounding the number of operations.}
Moreover, by Lemma \ref{Lemma_operation_count} we have that each iteration of Algorithm \ref{alg_Soft_Dikin_Walk} can be implemented in $ O(md^{\omega-1})$ arithmetic operations plus $O(1)$ calls to the oracle for the value of $f$.
 Thus, the number of iterations $T$ taken by Algorithm \ref{alg_Soft_Dikin_Walk}  is  $O((md + d L^2 R^2) \times \log(\frac{w}{\delta}))$ iterations in the setting where $f$ is $L$-Lipschitz and $O((md + d \beta R^2) \times \log(\frac{w}{\delta}))$ iterations in the setting where $f$ is $\beta$-smooth, where each iteration takes one function evaluation and $md^{\omega-1}$ arithmetic operations.

\end{proof}

\section*{Acknowledgments} 
The authors are grateful to Sushant Sachdeva and Yin Tat Lee for their insightful comments and suggestions.
This research was supported in part by  NSF grants CCF-1908347, CCF-2112665, and CCF-2104528.

\newpage

\bibliography{DP}
\bibliographystyle{plain}

\newpage

\appendix

\section{Proof of Lemma \ref{lemma_nu}} \label{sec_lemma_nu}

\begin{proof}[Proof of Lemma  \ref{lemma_nu}]
For any $x \in \mathrm{int}(K)$, we have
\begin{align*}
    \nabla g(x) = \nabla \phi(x) + \alpha x
\end{align*}
Thus, for any $h \in \mathbb{R}^d$,
\begin{equation} \label{eq_nu_1}
    h^\top \nabla g(x) = h^\top \nabla \phi(x) + \alpha h^\top x \leq  h^\top \nabla \phi(x)+ \alpha \|h\| \|x\| \leq  h^\top \nabla \phi(x)+ \alpha R\|h\|
\end{equation}
Thus,
\begin{align}\label{eq_nu_3}
    (h^\top \nabla g(x))^2 &\stackrel{\textrm{Eq. }\eqref{eq_nu_1}}{\leq}  (h^\top \nabla \phi(x)+ \alpha R\|h\|)^2\\
    &\leq 4(h^\top \nabla \phi(x))^2+ 4(\alpha R\|h\|)^2 \nonumber\\\nonumber
    &= 4(h^\top \nabla \phi(x))^2+ 4\alpha^2 R^2 h^\top I_d h \\\nonumber
        &= 4(h^\top \nabla \phi(x))^2+ 4\alpha R^2 h^\top (\alpha I_d) h \\\nonumber
    &\leq 4 \nu' h^\top \nabla^2 \phi(x) h + 4\alpha R^2 h^\top (\alpha I_d) h\\  \nonumber
    &\leq (4 \nu' + 4\alpha R^2)( h^\top \nabla^2 \phi(x) h + h^\top (\alpha I_d) h)\\ \nonumber
        &= (4 \nu' + 4\alpha R^2)( h^\top (\nabla^2 \phi(x) + \alpha I_d) h)\\ \nonumber
        &= (4 \nu' + 4\alpha R^2)( h^\top (\nabla^2 g(x)) h) \nonumber
\end{align}
where the third inequality holds since $\phi$ is $\nu'$ self-concordant.
Thus, plugging in $\nu = 4 \nu' + 4\alpha R^2$ to \eqref{eq_nu_3}, we get that
\begin{equation*}
    h^\top \nabla g(x) \leq \sqrt{\nu( h^\top (\nabla^2 g(x)) h)}.
\end{equation*}
\end{proof}

\section{Lower bounds for self-concordance parameter} \label{appendix_lower_bounds_nu}

\paragraph{Lower bounds on $\nu$ for worst-case $L$-Lipschitz $f$.}

Consider the $L$-Lipschitz function $f(\theta) = L \|\theta\|_2$ constrained to the convex body $K = \frac{R}{2 \sqrt{d}} [-1,1]^d$ which is contained in a ball of radius $R$.
In this case, any barrier function $g$ which satisfies Property \ref{property_ellipsoid} has Dikin Ellipsoid $E(\theta)$ which is contained in the ball $B(0,\frac{8}{L})$. 
Thus, we have that $\nu \geq \Omega(R L)$.  (This is because, by Proposition 2.3.2(iii) of \cite{nesterov1994interior}, any $\nu$-self concordant function $g$ satisfies $(h^\top \nabla^2 g(\theta) h)^{-\frac{1}{2}} \leq |h|_{\theta} \leq (1+ 3\nu) (h^\top \nabla^2 g(\theta) h)^{-\frac{1}{2}}$, for any $h \in \mathbb{R}^d$ where $|h|_{\theta} := \sup \{ \alpha>0 :   \theta \pm \alpha h \in  \{z \in K \}$. 
Thus, at $\theta = 0$ and choosing $h = (1,\cdots, 1)$ we have $\frac{|h|_\theta}{\|h\|_2} = R$ and $(h^\top \nabla^2 g(\theta) h)^{-\frac{1}{2}} \leq O(\frac{1}{L})$. 
Thus,  $\nu \geq \Omega\left (\frac{|h|_\theta}{(h^\top \nabla^2 g(\theta) h)^{-\frac{1}{2}}} \right) \geq L R$).
\\

Since there exists a convex body $K$ for which any self-concordant barrier function satisfying Definition \ref{def_barrier} has self-concordance parameter at least $\nu \geq d$, for any $L, R>0$ there exists a function $f$ and convex body $K \subset B(0,R)$ such that the self-concordance parameter of every barrier function satisfying both  Definition \ref{def_barrier}  and Property \ref{property_ellipsoid} is at least $\nu \geq \Omega(\max(d, L R))$.

\paragraph{Lower bounds on $\nu$ for worst-case $\beta$-smooth $f$.}
Consider the $\beta$-smooth function $f(\theta) = \frac{1}{2}\beta \theta^\top \theta$ constrained to the convex body $K = \frac{R}{2 \sqrt{d}} [-1,1]^d$.
At $\theta = 0$, any ellipsoid $D(\theta)$ satisfying Property \ref{property_ellipsoid} is contained in the ball  $\frac{8}{\sqrt{\beta}} B(0,1)$.
Thus, we have that $\nu \geq \frac{R}{\frac{8}{\sqrt{\beta}}} \sqrt{\beta} = \Omega(\sqrt{\beta} R)$ (This is because  at $\theta = 0$ we have $\frac{|h|_\theta}{\|h\|_2} =R$ and $(h^\top \nabla^2 g(\theta) h)^{-\frac{1}{2}} \leq O(\frac{1}{\sqrt{\beta}})$ and thus, by Proposition 2.3.2(iii) of \cite{nesterov1994interior}, we have $\nu \geq \Omega\left (\frac{|h|_\theta}{(h^\top \nabla^2 g(\theta) h)^{-\frac{1}{2}}} \right) \geq \sqrt{\beta} R$).  
Thus, for any $\beta, R>0$ there exists a function $f$ and convex body $K \subset B(0,R)$ such that the self-concordance parameter of every barrier function satisfying both  Definition \ref{def_barrier}  and Property \ref{property_ellipsoid} is at least $\nu \geq \max(d, \Omega(\beta R))$.

\end{document}